\documentclass[a4paper,11pt,english]{article}

\usepackage[utf8]{inputenc}
\usepackage[a4paper,bindingoffset=0cm,left=2.0cm,right=2.0cm,top=2.5cm,bottom=2.5cm,footskip=1.0cm]{geometry}

\usepackage[english]{babel}
\usepackage{verbatim}
\usepackage{tikz}
\usepackage{amssymb}
\usepackage{mathtools}
\usepackage{amsthm}
\usepackage{bm}
\usepackage{bbm}
\usepackage{multirow}
\usepackage{mleftright}\mleftright

\usepackage{qcircuit}
\usepackage{enumitem}
\usepackage{float}

\usepackage{titling}
\usepackage[final]{hyperref}
\hypersetup{
           breaklinks=true,   
           colorlinks=true,   
        }
\usepackage{cleveref}

\newcommand{\pvp}{\vec{p}{\kern 0.45mm}'}

\usepackage[linesnumbered,ruled,vlined]{algorithm2e}
\SetKwFor{RepTimes}{repeat}{times}{end}

\DeclarePairedDelimiter\bra{\langle}{\rvert}
\DeclarePairedDelimiter\ket{\lvert}{\rangle}
\DeclarePairedDelimiterX\braket[2]{\langle}{\rangle}{#1 \delimsize\vert #2}
\newcommand{\underflow}[2]{\underset{\kern-60mm \overbrace{#1} \kern-60mm}{#2}}
\newcommand{\vertiii}[1]{{\left\vert\kern-0.25ex\left\vert\kern-0.25ex\left\vert #1 
		\right\vert\kern-0.25ex\right\vert\kern-0.25ex\right\vert}}

\DeclareMathOperator{\spec}{spec}
\newcommand{\poly}{\operatorname{poly}}

\def\Ind(#1){{{\tt Ind}({#1})}}

\def\Pr{\mbox{Pr}}
\def\Tr{\mbox{Tr}}

\long\def\ignore#1{}

\newtheorem{theorem}{Theorem}
\newtheorem{corollary}[theorem]{Corollary}
\newtheorem{lemma}[theorem]{Lemma}
\newtheorem{prop}[theorem]{Proposition}
\newtheorem{definition}[theorem]{Definition}

\usepackage{ifdraft}
\usepackage{filecontents}
\ifdraft{\newcommand{\authnote}[3]{{\color{#3}({\bf  #1:} #2)}}}{\newcommand{\authnote}[3]{}}

\newcommand{\anote}[1]{\authnote{András}{#1}{red}}

\title{Improved Quantum Algorithms for Fidelity Estimation}
\author{
András Gilyén\thanks{Alfréd Rényi Institute of Mathematics, Hungary. Supported in part by the AWS Center for Quantum Computing. \texttt{gilyen@renyi.hu}} \thanks{Caltech, Pasadena, CA, USA. Funding provided by Samsung Electronics Co., Ltd., for the project ``The Computational Power of Sampling on Quantum Computers''. Additional support was provided by the Institute for Quantum Information and Matter, an NSF Physics Frontiers Center (NSF Grant PHY-1733907).} \quad\quad\quad\quad\,\,
Alexander Poremba\thanks{Caltech, Pasadena, CA, USA. Partially supported by AFOSR YIP award number number FA9550-16-1-0495 and the Institute for Quantum Information and Matter (an NSF Physics Frontiers Center; NSF Grant PHY-1733907).
 \texttt{aporemba@caltech.edu}}
}

\date{\vspace{5mm}\today\vspace{2mm}}

\begin{document}
	
	\maketitle
	
	\newcommand{\tr}[1]{\Tr\left[#1\right]}
	\newcommand{\sgn}[1]{\mathrm{sgn}\left(#1\right)}
	\newcommand{\rk}[1]{\mathrm{rk}\left(#1\right)}	
	\newcommand{\svt}[2]{\mathrm{SV}^{(#1)}\!(#2)}	
	\newcommand{\R}{\mathbb{R}}
	\newcommand{\C}{\mathbb{C}}		
	
	\newcommand{\eps}{\varepsilon}
	\newcommand{\bigket}[1]{\left |#1 \right \rangle}
	\newcommand{\bigbra}[1]{\left \langle#1 \right|}
	\newcommand{\ipc}[2]{\langle #1 , #2 \rangle}
	\newcommand{\ketbra}[2]{|#1\rangle\! \langle #2|}
	\renewcommand{\braket}[2]{\langle #1|#2 \rangle}
	\newcommand{\braketbra}[3]{\langle #1|#2| #3 \rangle}
	\newcommand{\nrm}[1]{\left\lVert #1 \right\rVert}
	\newcommand{\bigO}[1]{\mathcal{O}\left( #1 \right)}
	\newcommand{\bigOt}[1]{\widetilde{\mathcal{O}}\left( #1 \right)}
	\newcommand{\ctrlA}{\push{\rule{1.5mm}{1.5mm}}}
	\newcommand{\thPi}[2]{\Pi^{\phantom{#1}}_{#2}}	
	
	\begin{abstract}
		Fidelity is a fundamental measure for the closeness of two quantum states, which is important both from a theoretical and a practical point of view. Yet, in general, it is difficult to give good estimates of fidelity, especially when one works with mixed states over Hilbert spaces of very high dimension. Although, there has been some progress on fidelity estimation, all prior work either requires a large number of identical copies of the relevant states, or relies on unproven heuristics. 
		In this work, we improve on both of these aspects by developing new and efficient quantum algorithms for fidelity estimation with provable performance guarantees in case at least one of the states is \emph{approximately} low-rank. Our algorithms use advanced quantum linear algebra techniques, such as the quantum singular value transformation, as well as density matrix exponentiation and quantum spectral sampling. As a complementary result, we prove that fidelity estimation to any non-trivial constant additive accuracy is hard in general, by giving a sample complexity lower bound that depends polynomially on the dimension. Moreover, if circuit descriptions for the relevant states are provided, we show that the task is hard for the complexity class called \emph{(honest verifier) quantum statistical zero knowledge} via a reduction to a closely related result by Watrous.
	\end{abstract}
		
	\section{Introduction}	
	
	Today's quantum computers suffer from various kinds of incoherent noise (for example, as a result of $T_1$ and $T_2$ processes), which makes it difficult to use current quantum technologies to their best advantage~\cite{nielsen2002QCQI}. Characterizing  noise in quantum systems is therefore a fundamental problem for quantum computation and quantum information.
	Since quantum states have a much finer structure than their classical counterparts comprised of probability distributions, this calls for sophisticated distance measures between quantum states, such as \emph{trace distance}, the \emph{Bures metric} and \emph{fidelity}. 
	
	Each distance measure captures slightly different aspects of how two quantum states differ.
	While fidelity is not a \emph{metric} on the space of density matrices, it stands out by its versatility and applicability, and naturally appears in many practical scenarios.
	For example, it captures the geometric distance between thermal states of condensed matter systems nearing phase transitions, and can thus provide useful information about the zero temperature phase diagram~\cite{PhysRevA.75.032109,PhysRevE.79.031101}. In other contexts, the fidelity of quantum states allows one to infer chaotic behavior of thermofield dynamics of many-body quantum systems~\cite{PhysRevB.103.064309}.
	 
	The fidelity of two positive semi-definite operators $\rho$ and $\sigma$ on a Hilbert space\footnote{In the infinite dimensional setting one should also assume that $\rho$ and $\sigma$ are trace-class operators. To avoid similar difficulties in this paper we restrict our attention to the finite-dimensional case.} $\mathcal{H}$ is defined as
	\begin{align}\label{def:Fidelity}
	F(\rho,\sigma)
	=\tr{\sqrt{\sqrt{\rho}\sigma\sqrt{\rho}}}.
	\end{align}
	The fidelity is symmetric in $\rho$ and $\sigma$, and for quantum states its value lies between $0$ and $1$, equalling $1$ if and only if the states are identical. In this work, we are concerned with the problem of estimating the fidelity up to some additive error. In other words, given two density operators $\rho,\sigma \in \mathbb{C}^{d \times d}$ and $\eps \in (0,1)$, the problem is to output an additive $\eps$-approximation $\hat{F}(\rho,\sigma)$ such that
	\begin{align}
	F(\rho,\sigma) - \eps \,\leq\, \hat{F}(\rho,\sigma) \,\leq\, F(\rho,\sigma)  + \eps.
	\end{align}
	
	We study two input models. In the weaker input model called \emph{sampling access}, we only assume access to identical independent copies of the states, whereas in the stronger model called \emph{purified access}, we assume access to quantum circuits $U_\rho$ and $U_\sigma$ that allow one to prepare a purification of the quantum states. 
	In the latter model, we denote by $T_\rho$ and $T_\sigma$ the time complexity of preparing the purifications of $\rho$ and $\sigma$, respectively.
    Let us also denote by $r \in [d]$ the smallest rank of the two states. Without loss of generality, we can always assume that $r = \rk{\rho}$. 
    In general it is computationally difficult to give good estimates of the fidelity $F(\rho,\sigma)$, especially when one works with mixed states over Hilbert spaces of very high  dimension.
    
    In this work, we present new and efficient approximation algorithms for fidelity estimation which have $\poly(r,1/\eps)$ time and sample/query complexity, and far outperform previous algorithms in the literature. As a complementary result, we prove new hardness results and show that the task of approximating fidelity to any non-trivial constant error is hard for the complexity class called \emph{(honest verifier) quantum statistical zero knowledge}.
	
	\subsection{Related work} We now give an overview of approximation algorithms for fidelity estimation. 
	Let us first discuss the setting in which one of the states is pure, which is fairly well-understood as the fidelity reduces to the the simple quantity $F(\ket{\psi}\bra{\psi},\sigma) = \sqrt{\bra{\psi}\sigma\ket{\psi}}$. Buhrmann et al.~\cite{buhrman2001QuantumFingerprinting} gave an efficient quantum algorithm known as the \emph{Swap Test} which allows one to obtain an additive $\eps$-approximation  of $\bra{\psi}\sigma\ket{\psi}$ given $\tilde{O}(1/\eps^2)$ indentical copies of $\ket{\psi}$ and $\sigma$ (see also \cite{Cincio18}). Flammia and Liu~\cite{PhysRevLett.106.230501} subsequently gave a randomized $\eps$-approximation algorithm known as \emph{direct fidelity estimation} which only involves Pauli measurements on $\tilde{O}(d/\eps^2)$ samples of $\ket{\psi}$ and $\sigma$. The general task of fidelity estimation in which both density operators are mixed states is far less understood and requires a much more careful approach.
	
    A simple and direct method for estimating mixed state fidelity is through the use of quantum state tomography. O'Donnell and Wright~\cite{odonnell2016EfficientQuantumTomography} showed that, given $O(r \cdot d /\eps^2)$ copies of $\rho$, one can obtain an estimate $\hat{\rho}$ such that $\|\hat{\rho} - \rho\|_2 \leq \eps$. Using quantum state tomography, one can therefore construct a simple $\poly(d,1/\eps)$-time approximation algorithm which evaluates the fidelity $F(\hat{\rho},\hat{\sigma})$ directly given in the order of $\poly(d,1/\eps)$ many copies of $\rho$ and $\sigma$. 
    
    Cerezo et al.~\cite{cerezo2019VariaQuantFidEst} later studied the problem of low-rank fidelity estimation on near-term quantum computers via heuristic variational quantum algorithms that require many identical copies of $\rho$ and $\sigma$. In the same work, the authors also showed that low-rank fidelity estimation is hard for the complexity class called $\mathsf{DQC1}$, which consists of all problems that
    can be efficiently solved with bounded error in the \emph{one clean-qubit model} of quantum computation.
	Agarwal et al.~\cite{agarwal2021estimating} recently considered variational algorithms for fidelity estimation in other special cases. As a complementary result, the authors also showed that fidelity estimation in the case where one state is pure and the other is mixed is $\mathsf{BQP}$-complete. In the meantime Wang et al.~\cite{wang2021quantum} proposed a quantum algorithm in the purified access model that utilizes block-encoding techniques and computes an $\eps$-approximation to $F(\rho,\sigma)$ in time $O\left(\frac{r^{21.5}}{{\eps^{23.5}}}(T_\rho + T_\sigma)\right)$, where $T_\rho$ and $T_\sigma$ correspond to the time complexity of purified access, i.e., the complexity of the circuits  $U_\rho$ and $U_\sigma$ preparing purifications of $\rho$ and $\sigma$ respectively. Crucially, the work of Wang et al.~\cite{wang2021quantum} does not take the special case into consideration where one of the states is approximately low-rank.
	
	We give a summary of the most relevant results for fidelity estimation in the table below.
	
	\begin{center}
	\begin{tabular}{|p{7.5cm}||p{53mm}|p{2.8cm}|}
	\hline
	     Approximation method &Time/query/sample complexity &Assumptions\\
	     \hline
	     Quantum state tomography \cite{odonnell2016EfficientQuantumTomography} \vspace{2mm} & $\poly(d,1/\eps)$   &   identical copies\\
	     Variational fidelity estimation \cite{cerezo2019VariaQuantFidEst} \vspace{1.5mm} &N/A (heuristic) &  identical copies\\
	     Block-encoding algorithm \cite{wang2021quantum}
	     &   $\tilde{O}\left(\frac{r^{21.5}}{{\eps^{23.5}}}(T_\rho + T_\sigma)\right)$  &purified access\\[3mm]
	     \multirow{2}{*}{Our block-encoding algorithm (Section \ref{sec:block-encoding algorithm})}  & $\tilde{O}\left(\frac{r^{2.5}}{\eps^5}(T_\rho +T_\sigma) \right)$ &  purified access\\
	       & $\tilde{O}\left(\frac{r^{5.5}}{\eps^{12}} \right)$ &  identical copies\\[2.5mm]  
	     Our spectral sampling algorithm (Section \ref{sec:spectral-sampling-algorithm})& $
\tilde{O} \left( \frac{r^{10.5}(T_\rho  +T_\sigma)}{\eps^{25} \Delta }+ \frac{r^3 T_\rho }{\min \{\frac{\eps^7}{r^3} ,\Delta\}^3} \right)
	$  & purified access, spectrum*\\ 
	    Any (i.e., lower bound)~\cite{badescu2017QStateCertification,odonnell2015QuantumSpectrumTesting}& 
		$\Omega \left( \frac{r}{\eps} \right)$
		&  identical copies\\
	     \hline
	\end{tabular}
	\end{center}
	$^*$We remark that our spectral sampling-based algorithm assumes that $\rho$ has a $\Delta$-gapped spectrum. 
	
Our improved quantum
algorithms for fidelity estimation in Section \ref{sec:block-encoding algorithm} and Section \ref{sec:spectral-sampling-algorithm} achieve $\poly(r,1/\eps)$ time and sample/query complexity in order to output an $\eps$-estimate for the fidelity $F(\rho,\sigma)$. We remark that our algorithms give further significant improvements
in the case in which at least one of the states is \emph{approximately} low-rank through the use of truncation.\footnote{Just before submitting this manuscript we noticed the concurrent work of Wang et al.~\cite{wang2022quantum}. While our block-encoding based algorithm still has a far better complexity, our spectral sampling algorithm is less favourable in the worst case. However, it offers significant improvements in the approximate low-rank regime through the use of truncation.}

	\subsection{Our results}
	
	\paragraph{Fidelity estimation with block-encoding based algorithms.}
	Our first algorithm is based on advanced quantum linear algebra techniques, such as \emph{block-encodings} and the \emph{quantum singular value transformation} (QSVT)~\cite{gilyen2018QSingValTransf}. Our algorithm obtains an estimate $F(\rho_{\theta},\sigma)$, where $\rho_{\theta}$ is a so-called ``soft-thresholded'' version of $\rho$ in which eigenvalues of $\rho$ below $(1-\delta)\theta$ are completely removed and eigenvalues above $\theta$ are kept intact, while eigenvalues in the interval $[(1-\delta)\theta, \theta]$ are potentially missing or decreased by some amount. In the purified access model our algorithm has the complexity 
	\begin{equation*}
	\bigOt{\frac{\mathrm{rk}_\theta}{\eps^2\delta\theta^\frac{3}{2}}T_\rho+\frac{\mathrm{rk}_\theta^2}{\eps^4\theta^\frac{1}{2}}T_\sigma},
	\end{equation*}
	where $\mathrm{rk}_\theta$ is any upper bound on the number of eigenvalues of $\rho$ in the interval $[(1-\delta)\theta,1]$. Since $\mathrm{rk}_\theta\leq \frac{1}{(1-\delta)\theta}$, for any $\delta\in[0,\frac{1}{2}]$, the above can be always be bounded above by 
	\begin{equation*}
		\bigOt{\frac{T_\rho}{\eps^2\delta\theta^\frac{5}{2}}+\frac{T_\sigma}{\eps^4\theta^\frac{5}{2}}}.
	\end{equation*}
	On the other hand, if we know that $\rk{\rho}\leq r$, then choosing $\theta = \Theta(\frac{\eps^2}{r})$ and $\delta = \frac{1}{2}$ our algorithm obtains an $\eps$-precise estimate of $F(\rho,\sigma)$ in complexity
	\begin{equation*}
	\bigOt{\frac{r^\frac{5}{2}}{\eps^5}(T_\rho+T_\sigma)},
	\end{equation*}
	as we show that $|F(\rho,\sigma)-F(\rho_{\theta},\sigma)|\leq \sqrt{\tr{\thPi{\rho}{[0,\theta)} \rho \thPi{\rho}{[0,\theta)}}}$ in \Cref{subsec:trunFid}\footnote{Here $\Pi_{[0,\theta)}$ projects out the eigenvalues of $\rho$ below $\theta$ as defined in \Cref{cor:SoftBounds}.} analogously to~\cite{cerezo2019VariaQuantFidEst}.
	
	At its core, our algorithm builds on the \emph{Hadamard Test} which, given a quantum state $\rho$ and a block-encoding of a matrix $A$, outputs $0$ with probability $\mathrm{Re}(\tr{\rho (I+A)/2})$. Denoting by $U \Sigma V^\dagger$ a singular value decomposition of $\sqrt{\rho}\sqrt{\sigma}$, we can then write the fidelity as follows:
	\begin{align*}
	F(\rho,\sigma)=\nrm{\sqrt{\rho}\sqrt{\sigma}}_1
	=\tr{U^\dagger \sqrt{\rho}\sqrt{\sigma}V }
	=\tr{\sqrt{\rho}\sqrt{\sigma}V U^\dagger}
	=\tr{\rho (\rho^{-\frac{1}{2}}\sqrt{\sigma}V U^\dagger)}.
	\end{align*}
	Therefore, it suffices to construct a (subnormalized) block-encoding of $\rho^{-\frac{1}{2}}\sqrt{\sigma}V U^\dagger$ in order to compute an estimate of $F(\rho,\sigma)$. When working with $\rho_{\theta}$ instead of $\rho$ we can effectively bound $\nrm{\rho^{-\frac{1}{2}}}$ by $\bigO{\theta^{-\frac12}}$, so the block-encoding will have subnormlaization $\sim \theta^{-\frac12}$. This means that we can apply approximately $\frac{1}{\sqrt{\theta}\eps}$-rounds of amplitude amplification to obtain an $\eps$-precise estimate of $F(\rho_{\theta},\sigma)$ with high probability.
	
	In order to implement a block-encoding of $\rho^{-\frac{1}{2}}\sqrt{\sigma}V U^\dagger$ we implement each of $\rho^{-\frac{1}{2}}$, $\sqrt{\sigma}$, and $V U^\dagger$ as individual block-encodings and simply take the products of the block-encodings, which is a native quantum operation~\cite{gilyen2018QSingValTransf}. Here, a key observation is that the purified access model implies that we also have a unitary block-encoding of $\rho$ and $\sigma$~\cite{gilyen2018QSingValTransf}, so we can obtain block-encodings of $\rho^{-\frac{1}{2}}$ and $\sqrt{\sigma}$ by applying the QSVT on the block-encodings of $\rho$ and $\sigma$. We can obtain an approximate implementation of $V U^\dagger$ by implementing an (approximate) block-encoding of $\sqrt{\sigma}\sqrt{\rho}$ then applying ``singular vector transformation''~\cite{gilyen2018QSingValTransf}. Again in order to implement a block-encoding of $\sqrt{\sigma}\sqrt{\rho}$ we can simply implement block-encodings of both $\sqrt{\sigma}$ and $\sqrt{\rho}$ via the QSVT. 
	On a high level, the above describes the essence of our block-encoding-based algorithm. Since the QSVT only allows for polynomial transformations, we need to give appropriate polynomial approximations of $x^{\pm\frac12}$ -- the details of the approximation error and complexity analysis can be found in \Cref{sec:block-encoding algorithm}.
	
	In case we only have access to samples of $\rho$ and $\sigma$, we can still use \emph{density matrix exponentiation}~\cite{lloyd2013QPrincipalCompAnal,kimmel2016hamiltonian} in combination with the QSVT to implement approximate block-encodings of $\rho$ and $\sigma$, and use essentially the same strategy as described above. However, in order to maintain the required accuracy throughout the circuit our algorithm requires the use of a large number of samples.

	\paragraph{Fidelity estimation via quantum spectral sampling.} Our
	second approximation algorithm for fidelity estimation exploits the fact that it is possible to ``sample'' from the spectrum of density operators.
	The main idea is the following. Suppose we wish to estimate the fidelity $F(\rho,\sigma)$, for density matrices $\rho,\sigma \in \mathbb{C}^{d \times d}$. Let $r = \rk{\rho}$ be the rank of $\rho$, and let $\spec(\rho) = (\lambda_1, \dots, \lambda_r)$ be the spectrum of $\rho$ (with multiplicity). Expanding $\sigma$ in the eigenbasis of $\rho=\sum_{i=1}^{r} \lambda_i \ket{\psi_i}\bra{\psi_i}$, we find that
	\begin{align}
	 F(\rho,\sigma) 
	 = \Tr{\left[ \sqrt{\sqrt{\rho}\sigma\sqrt{\rho}} \right]} = \Tr{\left[ \sqrt{\sum_{i,j \in [r]} \sqrt{\lambda_i}\sqrt{\lambda_j}  \bra{\psi_i} \sigma \ket{\psi_j} \,\, \ket{\psi_i}\bra{\psi_j}}\right]}.
	\end{align}
	In other words, we can write the fidelity between $\rho$ and $\sigma$ as the quantity $F(\rho,\sigma) = \Tr{[\sqrt{\Lambda}]}$, where
	$\Lambda(\rho,\sigma)=\sqrt{\rho}\sigma\sqrt{\rho} \in \mathbb{C}^{d \times d}$ has the following non-trivial entries in the eigenbasis of $\rho$:
	\begin{align*}
	 \Lambda_{ij}  = \sqrt{\lambda_i}\sqrt{\lambda_j}  \bra{\psi_i} \sigma \ket{\psi_j}, \quad\quad \forall i,j \in \{1,\dots,r\}.
	\end{align*}
	Hence, it suffices to directly compute the matrix elements of $\Lambda$ in order to estimate the fidelity $F(\rho,\sigma)$.
	Let us first consider the case of \emph{exact} fidelity estimation in order to illustrate how our spectral sampling algorithm works. We remark that our spectral sampling algorithm can handle the case when $\rho$ is approximately low-rank using a soft-thresholding approach similar to our block-encoding algorithm.
	
	To estimate the eigenvalues of $\rho$, we use the idea of \emph{quantum spectral sampling} first introduced by Lloyd, Mohseni and Rebentrost~\cite{lloyd2013QPrincipalCompAnal} in the context of \emph{quantum principal component analysis}, and later extended by Prakash \cite{prakash2014QLinAlgAndMLThesis}. This subroutine allows us to approximately perform quantum phase estimation on $\rho$ with respect to a unitary $e^{- 2 \pi i \rho}$, resulting in the operation
	\begin{align}\label{eq:QPE-sample}
	\rho=\sum_{i=1}^{r} \lambda_i \ket{\psi_i}\bra{\psi_i}   \quad \mapsto \quad \sum_{i=1}^{r} \lambda_i \, \ket{\psi_i}\bra{\psi_i} \otimes \ket{\tilde \lambda_i}\bra{\tilde \lambda_i}.
	\end{align}
	By repeatedly performing the operation in \eqref{eq:QPE-sample}, we can sample random pairs of eigenstates and eigenvalues $(\ket{\psi_j},\tilde{\lambda}_j)$, where $\tilde{\lambda}_j \approx \lambda_j$. In order to obtain a full collection of all eigenvalues of $\spec(\lambda)$, we have to repeat this procedure multiple times, which raises the question: How many repetitions of the quantum spectral sampling procedure are necessary to find a full collection of $r$ distinct eigenvalues? To distinguish between the different eigenvalues, we must assume that $\rho$ has a non-degenerate spectrum $\spec(\rho) = (\lambda_1, \dots, \lambda_r)$, where each eigenvalue is separated by a gap $\Delta >0$ with
	\begin{align}
	\Delta = \min_{i \in [r-1]} \big|\lambda_{i+1} - \lambda_i \big|.
	\end{align}
	In Section \ref{sec:CCP}, we give concrete upper bounds on the number of repetitions needed to complete a full collection. In particular, we analyze the \emph{non-uniform coupon collector problem} which asks how many draws are needed to collect all $r$ eigenvalues, where the $i$-th eigenvalue is drawn with probability $\lambda_i \in (0,1]$. Denoting by $T_{\spec(\rho)}$ the random variable for the number of draws needed to complete the collection, we have by an identity due to Flajolet et al.~\cite{FLAJOLET1992207},
\begin{align}
\mathbb{E}[T_{\spec(\rho)}] = \int_{0}^{\infty} \Big( 1 - \prod_{i=1}^r (1 - e^{- \lambda_i t}) \Big) \,\mathrm{d}t.
\end{align}
Our first result is a non-trivial upper bound on the average number of draws in the non-uniform coupon collector problem. In Lemma \ref{lem:non-trivial-CCP}, we show that
\begin{align}\label{eq:non-trivial-bound}
\mathbb{E}[T_{\spec(\rho)}] 
\, \leq \, r \cdot H(\spec(\rho))^{-1},
\end{align}
where $H(x) = r/\sum_{i=1}^r x_i^{-1}$ is the harmonic mean of $x = (x_1, \dots, x_r)$. This allows us to directly relate the average number of draws necessary to complete the collection to spectral properties of $\rho$. Unfortunately, our initial bound in \eqref{eq:non-trivial-bound} is not tight. In particular, for the uniform spectrum $(\frac{1}{r},\dots,\frac{1}{r})$, our bound tells us that $\mathbb{E}[T_{(\frac{1}{r},\dots,\frac{1}{r})}]\leq r^2$, whereas a well known result on the (standard) \emph{uniform coupon collector problem} states that the average number of draws is in the order of $\Theta(r \log r)$. 

In order to further improve on the bound in \eqref{eq:non-trivial-bound}, we use a \emph{coupling argument} which allows us to relate instances of the non-uniform coupon collector problem to \emph{worst-case} instances of the uniform coupon collector problem. For example, we show in Lemma \ref{lem:couping-coupon-lemma} that, if $\kappa \in (0,1)$ is a lower bound on the smallest eigenvalue of $\rho$, then it holds that
$$
\mathbb{E}[T_{\spec(\rho)}] \,\leq\,  \mathbb{E}[T_{(\frac{1}{m},\dots,\frac{1}{m})}] = \Theta\left(\log(1/\kappa)/\kappa\right),
$$
where we choose $m = \lceil \frac{1}{2 \kappa}\rceil$. In \Cref{cor:threshold-CC}, we generalize the former by introducing a threshold parameter $\theta \in (0,1)$ and only considering eigenvalues of $\rho$ which lie above $\theta$. This allows us to obtain upper bounds that asymptotically match the bounds for the uniform coupon collector problem.

Going back to fidelity estimation, let us now describe how we can approximate $\bra{\psi_i} \sigma \ket{\psi_j}$, which is an additional quantity required to estimate the matrix elements $\Lambda_{ij}$, for all $i,j \in [r]$. Our first observation is that the diagonal entries of $\Lambda$ can easily be estimated via the \emph{Swap Test} introduced by Buhrmann et al.~\cite{buhrman2001QuantumFingerprinting}. In particular, once we have obtained a complete collection of pairs of eigenstates and eigenvalues $(\ket{\psi_i},\tilde{\lambda}_i)$, we can use the Swap Test on input $\ketbra{\psi_i}{\psi_i}$ and $\sigma$ to estimate $\bra{\psi_i} \sigma \ket{\psi_i}$ up to inverse-polynomial (in $\log d$) additive error.
Unfortunately, estimating the off-diagonal entries of the matrix $\Lambda$ is a lot more involved, since $\sigma_{ij} = \bra{\psi_i} \sigma \ket{\psi_j}$ is, in general, a complex number which contains both a real and an imaginary part. 

One possible solution for estimating $\sigma_{ij}$ is to use \emph{density matrix exponentiation} introduced by Lloyd, Mohseni and Rebentrost~\cite{lloyd2013QPrincipalCompAnal} which allows us to approximately implement a unitary $e^{- i \sigma t}$, for small $t \in (0,1)$. Let $j \in [r]$ be an index. A second-order Taylor expansion of $e^{- i \sigma t}$ reveals that
\begin{align}
e^{- i \sigma t} \ket{\psi_j} = \ket{\psi_j} - i t \,\sigma \ket{\psi_j} + O(t^2).
\end{align}
Therefore, for any index $i\in [r]$, we obtain the following identity,
\begin{align}\label{eq:taylor-sandwich}
\bra{\psi_i} e^{- i \sigma t} \ket{\psi_j} = \braket{\psi_i}{ \psi_j} - i t \,\bra{\psi_i}\sigma \ket{\psi_j} + O(t^2).
\end{align}
Re-arranging the quantity in \eqref{eq:taylor-sandwich}, we find that
\begin{align}\label{eq:identity-for-sigma-ij}
\bra{\psi_i} \sigma \ket{\psi_j} = \frac{-i}{t} \cdot (\braket{\psi_i}{ \psi_j} - \bra{\psi_i} e^{- i \sigma t} \ket{\psi_j}) + O(t),
\end{align}
which yields an approximate formula for $\sigma_{ij} = \bra{\psi_i} \sigma \ket{\psi_j}$ (up to first order in $t$). Notice that the right-hand side of Eq.~\eqref{eq:identity-for-sigma-ij} consists of simple overlaps between pure states. Unfortunately, we cannot apply the Swap Test to estimate the above quantities, since we are dealing with complex-valued inner products. Hence, we must rely on the so-called \emph{Hadamard Test} due to Aharanov, Jones and Landau~\cite{aharonov2006PolynomialQAlgForJonesPoly}, which allows one to estimate the real and imaginary parts of $\bra{\psi}U\ket{\psi}$, for a state $\ket{\psi}$ and unitary $U$.
However, in order to estimate the required quantities in Eq.~\eqref{eq:identity-for-sigma-ij}, we have to make use of an additional technique. Namely, we use \emph{quantum eigenstate filtering} due to Lin and Tong~\cite{lin2019OptimalQEigenstateFiltering} in order to approximately obtain circuits $U_i$ (and $U_i^\dag$) that prepare (and uncompute) eigenstates $\ket{\psi_i}$ of the state $\rho$ via purified access to $U_\rho$, for every index $i \in [r]$. This allows us to estimate $\bra{\psi_i} \sigma \ket{\psi_j}$ by instead approximating the following simple quantities up to inverse-polynomial  (in $\log d$) precision via the Hadamard Test:
\begin{align*}
\braket{\psi_i}{ \psi_j} &=  \mathrm{Re}{\bra{\psi_i} U_j U_i^\dag \ket{\psi_i}} + \mathrm{Im}{\bra{\psi_i} U_j U_i^\dag \ket{\psi_i}} \vspace{4mm}\\
\bra{\psi_i} e^{- i \sigma t} \ket{\psi_j} &= \mathrm{Re}{\bra{\psi_i} e^{- i \sigma t}U_j U_i^\dag \ket{\psi_i}} + \mathrm{Im}{\bra{\psi_i} e^{- i \sigma t}U_j U_i^\dag \ket{\psi_i}}.
\end{align*}

Another possible -- and much more efficient -- solution for estimating the quantity $\sigma_{ij} = \bra{\psi_i} \sigma \ket{\psi_j}$ is the following. Rather than using density matrix exponentiation, our spectral sampling algorithm implements a block-encoding of $\sigma$ which we can easily construct via \emph{purified access} to the state $\sigma$. Letting $U$ denote the associated block-encoding unitary, we perform a Hadamard test with respect to $\ket{0}\bra{0}$ and $U_j^\dag U U_i$ to directly estimate the real and imaginary parts of
$$
\mathrm{tr}[\ket{0}\bra{0} U_i^\dag U U_j ] =\bra{\psi_i} \sigma \ket{\psi_j}=\sigma_{ij}, \quad \forall i,j \in [r].
$$

Therefore, we can estimate the matrix entries $\Lambda_{ij}  = \sqrt{\lambda_i}\sqrt{\lambda_j}  \bra{\psi_i} \sigma \ket{\psi_j}$, for every pair of indices $i,j \in [r]$. Denoting our estimate by $\hat{\Lambda}$, we can then obtain an approximate fidelity estimate by computing $\hat{F}(\rho,\sigma) = \Tr{[\sqrt{\hat{\Lambda}_+}}]$, where $\hat{\Lambda}_+$ is the projection of $\hat{\Lambda}$ onto the positive semidefinite cone. 
We show in in Theorem \ref{thm:spectral-sampling-theorem} that our spectral sampling algorithm obtains an $\eps$-estimate $F(\rho_{\theta},\sigma)$ with high probability, where $\rho_{\theta}$ is a ``soft-thresholded'' version of $\rho$ with $\theta \in (0,1)$, in time 
$$
\tilde{O} \left( \frac{T_\rho  +T_\sigma}{\theta^{10.5}\eps^{4} \Delta } + \frac{T_\rho }{\theta^3 \min \{ \theta^3 \eps ,\Delta\}^3} \right).
$$
Finally, if we know that the rank of $\rho$ is at most $\rk{\rho}\leq r$, then choosing $\theta =\Theta(\frac{\eps^2}{r})$ we obtain an $\eps$-precise estimate of $F(\rho,\sigma)$ with high probability in time
$$
\tilde{O} \left( \frac{r^{10.5}(T_\rho  +T_\sigma\big)}{\eps^{25} \Delta }+ \frac{r^3 T_\rho }{\min \{\frac{\eps^7}{r^3} ,\Delta\}^3} \right).
$$
While our spectral-sampling based algorithm for fidelity estimation performs significantly worse than our block-encoding algorithm, it may be easier to implement in certain settings; for example, when it is easy to obtain circuits that prepare the eigenstates of one of the density operators.

	\subsection{\texorpdfstring{$\mathsf{QSZK_{HV}}$}{QSZK\textunderscore(HV)}-hardness of fidelity estimation to any non-trivial accuracy}

	Now we show that fidelity estimation to any non-trivial fixed precision is $\mathsf{QSZK_{HV}}$-hard. This provides evidence for the intractability of the problem in general without further assumptions on the states.
	
	\begin{theorem}[$\mathsf{QSZK_{HV}}$-hardness of non-trivial fidelity estimation]
		Consider the following problem: one is given (the description) of two quantum circuits $U,V$ preparing purifications of quantum states $\rho$ and $\sigma$ respectively, and the task is to output a number $\hat{F}(\rho,\sigma)$ such that $|\hat{F}(\rho,\sigma)-F(\rho,\sigma)|\leq \frac{1}{2}-\delta$. This problem is $\mathsf{QSZK_{HV}}$-hard for every $\delta \in (0,\frac{1}{2}]$.
	\end{theorem}
	\begin{proof}
	By the Fuchs--van de Graaf inequalities, we have 
	\begin{equation}\label{eq:FuchsGraaf}
	1- F(\rho,\sigma) \leq \frac{1}{2}\nrm{\rho-\sigma}_1\leq \sqrt{1- F(\rho,\sigma)^2}.
	\end{equation}
	Suppose we are given quantum circuits preparing purifications of $\rho$ and $\sigma$ and we are promised that either $\frac{1}{2}\nrm{\rho-\sigma}_1 \leq \eps$ or $\frac{1}{2}\nrm{\rho-\sigma}_1 \geq 1 - \eps$ for some constant $\eps\in(0,\frac{1}{8}]$. Watrous proved that this problem is $\mathsf{QSZK_{HV}}$-complete~\cite{watrous2003LimitationsOnQSZK}. By \Cref{eq:FuchsGraaf} $\frac{1}{2}\nrm{\rho-\sigma}_1 \leq \eps$ implies $1 - \eps \leq F(\rho,\sigma)$, and $1- \eps \leq \frac{1}{2}\nrm{\rho-\sigma}_1$ implies $F(\rho,\sigma)\leq \sqrt{2\eps}$. In particular estimating the fidelity $F(\rho,\sigma)$ to precision $\frac{1}{2}-\sqrt{2\eps}$ solves the distinguishing problem. Substituting $\delta:=\sqrt{2\eps}$ this implies that for every $\delta \in (0,\frac{1}{2}]$ fidelity estimation to precision $\frac{1}{2}-\delta$ is $\mathsf{QSZK_{HV}}$-hard. Since estimating the fidelity to precision $\frac{1}{2}$ is trivial (taking estimate $\frac{1}{2}$), this means that fidelity estimation to any fixed non-trivial accuracy is $\mathsf{QSZK_{HV}}$-hard in general.
	\end{proof}	
	
	\subsection{A sample complexity lower bound for constant precision fidelity estimation}
	
	Now we prove that any non-trivial fidelity estimation algorithm must use at least a polynomially large number of copies even if one of the states is known in advance.
	
	As B{\u{a}}descu, O'Donnell, and Wright~\cite{badescu2017QStateCertification} pointed out testing closeness with respect to fidelity requires a number of copies of the states proportional to the dimension of the states even if one of the states is a fixed known state, namely the completely mixed state. Their observation follows from a reduction to the earlier results of O'Donnell, and Wright~\cite{odonnell2015QuantumSpectrumTesting}.
	
	\begin{corollary}
		Let $\delta \in (0,1/7]$, and consider the following problem: Given a known quantum state $\sigma$ of rank $r$ and copies of a state $\rho$ with the promise that $\rk{\rho}\leq r$, then computing an estimate $\hat{F}(\rho,\sigma)$ such that $|\hat{F}(\rho,\sigma)-F(\rho,\sigma)|\leq \delta$ requires using $\Omega(r/\delta)$ copies in general.
	\end{corollary}
	\begin{proof}
		We proceed by a reduction to \cite[Theorem 1.7]{badescu2017QStateCertification}, which considers $\sigma$ to be the completely mixed state in dimension $r$ and $\rho$ to be an arbitrary state having half of its eigenvalues $\frac{1-\eps}{r}$ and the other half $\frac{1-\eps}{r}$. As~\cite{badescu2017QStateCertification} notes it follows from the results of \cite{odonnell2015QuantumSpectrumTesting} that distinguishing $\sigma$ from such states $\rho$ requires using $\Omega(\frac{r}{\eps^2})$ samples for every $\eps\in [0,1]$. Although they state the result in term of the dimensionality of the Hilbert space, adding extra dimension to the Hilbert space will not reduce the sample complexity, so this result can also be stated in terms of rank. 
		
		On the other hand a $\delta:=\frac{F(\sigma,\sigma)-F(\rho,\sigma)}{2}=\frac{1-F(\rho,\sigma)}{2}$-precise fidelity estimation algorithm can in particular distinguish $\sigma$ and $\rho$, and thereby any such algorithm must use at least $\Omega(\frac{r}{\eps^2})=\Omega(\frac{r}{\delta})$ samples, since  $F(\rho,\sigma)=\frac{1}{2}(\sqrt{1+\eps}+\sqrt{1-\eps})\leq 1-\frac{\eps^2}{8}$ for every $\eps\in[0,1] \Rightarrow \delta\geq \frac{\eps^2}{16}$.\footnote{The Taylor series of $f(x):=\sqrt{1+x}$ is $1+\frac{x}{2}-\frac{x^2}{8}+\frac{x^3}{16}-\bigO{x^4}$. By Lagrange's remainder theorem we get that for every $x\in(-1,1)$ we have $f(x)= 1+\frac{x}{2}-\frac{x^2}{8}+\frac{x^3}{16}+\frac{f^{(4)}(\eta)}{4!}\eta^4$ for some $\eta\in[-|x|,|x|]$. Since the fourth derivative of $f(x)=\sqrt{1+x}$ is $f^{(4)}(x)=\frac{15}{16 (x+1)^{7/2}}$ which is non-negative on $(-1,1)$ we get the inequality $\sqrt{1+x}\leq 1+\frac{x}{2}-\frac{x^2}{8}+\frac{x^3}{16}$ for every $x\in[-1,1]$.}
	\end{proof}

	\begin{prop}
		Let $\delta \in (0,\frac{1}{2})$, and consider the following problem: Given a known quantum state $\sigma$ of rank $r$ and copies of a state $\rho$ with the promise that $\rk{\rho}\leq r$, then computing an estimate $\hat{F}(\rho,\sigma)$ such that $|\hat{F}(\rho,\sigma)-F(\rho,\sigma)| < \frac12-\delta$ requires using $\Omega(\delta\sqrt{r})$ copies in general.
	\end{prop}
	\begin{proof}
		Consider $\sigma$ to be the uniform distribution over $[r]$, and the set of states $\rho$ that are uniform over a $\delta^2 r$-sized subset of $[r]$. Then $F(\sigma,\rho)=\delta$. Suppose we either get copies from $\sigma$ or a uniformly random $\rho$. Until an element is repeated all that we see are distinct uniformly random elements of $[0,r]$ and in order to find a repetition with non-negligible probability we need to obtain at least $\Omega(\delta r)$ samples.
		On the other hand estimating fidelity to precision better than $\frac{1}{2}-\frac{\delta}{2}$ can distinguish the two cases.
	\end{proof}

	We think that the above bound can be improved to $\Omega(r)$ using the techniques of~\cite{odonnell2015QuantumSpectrumTesting}, when one considers the set of states $\rho$ that are uniform over a $\delta^2 r$-dimensional subspace of the support of $\sigma$, however this result does not seem to directly follow from the results of~\cite{odonnell2015QuantumSpectrumTesting}. \anote{I am not sure about this, double check!}
	
	\section{Preliminaries}
	
	For a matrix $A \in \mathbb{C}^{m\times n}$ and $p\in[1,\infty]$, we denote by $\nrm{A}_p$ the \emph{Schatten $p$-norm}, which is the $\ell^p$-norm of the singular values $(\sum_i \varsigma_i^p(A))^{1/p}$.
	In particular, we use the notation $\|A\| =\|A\|_\infty$.
	We recall some useful inequalities~\cite[Section IV.2]{bhatia1997MatrixAnalysis}.
	\emph{Hölder's inequality} states that for all $B\in\mathbb{C}^{n\times k}$ and $r,p,q\in[1,\infty]$ such that $\frac1p+\frac1q=\frac1r$, we have $\nrm{AB}_r\leq\nrm{A}_p\nrm{B}_q$.\footnote{Note that the expression $(\sum_i \varsigma_i^p(A))^{1/p}$ makes sense for every $p>0$, but will not give a norm for $p\in (0,1)$ (due to violating the triangle inequality). Nevertheless, Hölder's inequality holds for these quantities as well, which can be formulated as follows \cite[Exercise IV.2.7]{bhatia1997MatrixAnalysis}: $\nrm{|AB|^r}_1^{\frac{1}{r}}\leq\nrm{|A|^p}_1^{\frac{1}{p}}\nrm{|B|^q}_1^{\frac{1}{q}}$ for all $r,p,q\in(0,\infty]$ such that $\frac1p+\frac1q=\frac1r$, where $|X|=\sqrt{X^\dagger X}$.}
	The \emph{trace-norm inequality} states that if $n=m$, then $|\Tr(A)|\leq\nrm{A}_1$.
	For a hermitian matrix $A \in \mathbb{C}^{n \times n}$ with spectral decomposition $A = U D U^\dag$ and $D=\mathrm{diag}(\lambda_1,\dots,\lambda_n)$, we denote by $A_+ = U D_+ U^\dag$ and $A_- = U D_- U^\dag$ the projections onto the positive semidefinite and negative semidefinite cone, respectively, where we let $D_+ = \mathrm{diag}(\max\{0,\lambda_1\},\dots,\max\{0,\lambda_n\})$ and $D_- = \mathrm{diag}(\min\{0,\lambda_1\},\dots,\min\{0,\lambda_n\})$.
	For an integer $m \in \mathbb{N}$, we denote by $H_m$ the $m$-th harmonic number given by $H_m = 1 + \frac{1}{2} + \frac{1}{3} + \dots + \frac{1}{m}$.	
	
	For a function $f\colon \R \mapsto \C$ and set $S\subseteq \R$ we use the notation $\nrm{f}_S:=\sup_{x\in S}|f(x)|$.
	
	\begin{definition}[Purified access]	Let $\rho \in \mathbb{C}^{d \times d}$ be a density operator. We say that we have purified access to the state $\rho$ if we have access to a unitary $U_\rho$ (and its inverse) acting as follows:
	$$
	U_\rho \ket{0}_A \ket{0}_B = \ket{\psi_\rho}_{AB} = \sum_{i=1}^d \lambda_i \ket{\phi_i}_A \ket{\psi_i}_B,
	$$
	where $\Tr_A \left[\ket{\psi_\rho}\bra{\psi_\rho}_{AB} \right] = \rho$, and where it holds that $\braket{\phi_i}{\phi_j} = \braket{\psi_i}{\psi_j} = \delta_{ij}$, for all $i,j \in [d]$. In this context, we denote by $T_\rho$ the time it takes to implement the unitary $U_\rho$.
	\end{definition}
	
	We use the following result which is a slight adaptation of \cite[Fact 7.4.9.2]{0521386322}.
	
	\begin{lemma}[Projection onto the positive semidefinite cone]\label{lem:projection_cone} Let $A \in \mathbb{C}^{n \times n}$ be a hermitian matrix
	with spectral decomposition $A = U D U^\dag$, where $D=\mathrm{diag}(\lambda_1,\dots,\lambda_n)$, and let $A_+ = U D_+ U^\dag$ be the projection onto the positive semidefinite cone with spectrum $D_+ = \mathrm{diag}(\max\{0,\lambda_1\},\dots,\max\{0,\lambda_n\})$.
	Let $\vertiii{\cdot}$ be any unitarily invariant norm over $\mathbb{C}^{n \times n}$. Then, it holds that
	$$
	A_+ = \underset{X \succeq 0}{\mathrm{argmin}} \vertiii{ A - X }.
	$$
	In other words, $A_+ $ is the closest positive semidefinite matrix to $A$ with respect to the norm $\vertiii{\cdot}$.
	\end{lemma}
	
	\subsection{Matrix Arithmetics using blocks of unitaries}\label{sec:matArith}
	
	In this section we recall some basic results from the generic matrix arithmetic toolbox described in \cite{gilyen2018QSingValTransfArXiv}, which is a distilled version of the results of a series of works on quantum algorithms \cite{harrow2009QLinSysSolver,berry2014HamSimTaylor,childs2015QLinSysExpPrec,low2016HamSimQubitization,apeldoorn2017QSDPSolvers,chakraborty2018BlockMatrixPowers}.
	
	First we introduce the definition of block-encoding which the main idea of which is to represents a subnormalized matrix as the upper-left block of a unitary.
	$$ U=\left[\begin{array}{cc} A/\alpha & . \\ . & .\end{array}\right] \kern10mm\Longrightarrow\kern10mm A =\alpha (\bra{0}\otimes I)U(\ket{0}\otimes I)$$
	\begin{definition}[Block-encoding]\label{def:standardForm}
		Suppose that $A$ is an $s$-qubit operator, $\alpha,\eps\in\R_+$ and $a\in \mathbb{N}$, then we say that the $(s+a)$-qubit unitary $U$ is an $(\alpha,a,\eps)$-block-encoding of $A$, if 
		$$ \nrm{A - \alpha(\bra{0}^{\otimes a}\otimes I)U(\ket{0}^{\otimes a}\otimes I)}\leq \eps. $$
	\end{definition}\noindent
	In case $\alpha=1$ and $\eps=0$ we simply call an $(\alpha,a,\eps)$-block-encoding a block-encoding (with $a$ ancillas).
	
	There are several ways to construct block-encodings, for a summary of the techniques we refer to~\cite{gilyen2018QSingValTransf}. For our work the most important is the following result due to Low and Chuang~\cite{low2016HamSimQubitization}:
	\begin{lemma}[Block-encoding of density operators with purified access {\cite[Lemma 45]{gilyen2018QSingValTransfArXiv}}]\label{lem:blochDensity}
		Suppose that $\rho$ is an $s$-qubit density operator and $G$ is an $(a+s)$-qubit unitary that on the $\ket{0}\ket{0}$ input state prepares a purification $\ket{0}\ket{0}\rightarrow \ket{\rho}$, s.t. $\mathrm{Tr}_{a}{\ketbra{\rho}{\rho}}=\rho$. Then $(G^\dagger\otimes I_s)(I_a\otimes \mathrm{SWAP}_s)(G\otimes I_s)$ is a $(1,a+s,0)$-block-encoding of $\rho$.
	\end{lemma}
	
	Block-encodings are convenient to work with, in particular one can efficiently construct linear combinations of block-encodings via the the so-called linear combination of unitaries technique. Moreover, one can also easily form products of block-encodings as follows:
	
	\begin{lemma}[Product of block-encoded matrices {\cite[Lemma 53]{gilyen2018QSingValTransfArXiv}}]\label{lemma:disjointAncillaProduct}
		If $U$ is an $(\alpha,a,\delta)$-block-encoding of an $s$-qubit operator $A$, and $V$ is an $(\beta,b,\eps)$-block-encoding of an $s$-qubit operator $B$ then\footnote{The identity operators act on each others ancilla qubits, which is hard to express properly using simple tensor notation, but the reader should read this tensor product this way.} $(I_b\otimes U)(I_a\otimes V)$ is an $(\alpha\beta,a+b,\alpha\eps+\beta\delta)$-block-encoding of $AB$.
	\end{lemma}
	
	\subsection{The Swap Test}\label{subsec:swap-test}
	
	Let us now recall the \emph{Swap Test} introduced by Buhrmann et al.~\cite{buhrman2001QuantumFingerprinting}. We remark that detailed circuits for the general case can also be found in the work of Cincio et al.~\cite{Cincio18}.
	Given as input identical copies of density operators $\ket{\psi}\bra{\psi}$ and $\sigma$, we can repeat the following quantum circuit 
		
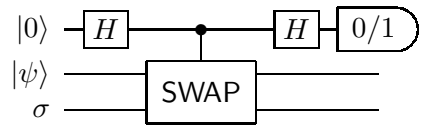
\begin{displaymath}
		\Qcircuit @C=2.47mm @R=1.32mm {
			\lstick{\ket{0}} & \gate{H} & \ctrl{1} & \gate{H} & \measureD{0/1} \\
			\lstick{\ket{\psi}}& \qw & 	 \multigate{1}{\mathsf{SWAP}} & \qw  & \qw\\
			\lstick{\sigma}& \qw &  \ghost{\mathsf{SWAP}} & \qw  & \qw
		}		
		\end{displaymath} 
	$\log\big( \frac{2}{\nu} \big)/{\eta^2}$ many times with parameters $\eta \in (0,1)$ and $\nu \in (0,1)$ to obtain an additive $\eta$-approximation to the quantity $\bra{\psi}\sigma \ket{\psi}$ with probability at least $1-\nu$.
	
	\subsection{The Hadamard test and block-measurements}\label{subsec:HadTest}
	
	In this section, we recall the \emph{Hadamard Test} due to Aharanov, Jones and Landau~\cite{aharonov2006PolynomialQAlgForJonesPoly}.
	Given as input identical copies of a state $\ket{\psi}$ and a unitary $U$, we can repeat the following circuit
	
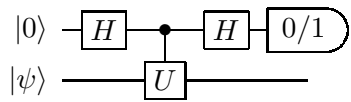
\begin{displaymath}
		\Qcircuit @C=2.47mm @R=1.32mm {
			\lstick{\ket{0}} & \gate{H} & \ctrl{1} & \gate{H} & \measureD{0/1} \\
			\lstick{\ket{\psi}}& \qw & 	\gate{U} & \qw  & \qw	
		}		
		\end{displaymath} 	
	$4 \log\big( \frac{4}{\delta} \big)/\xi^2$ times with parameters $\xi \in (0,1)$ and $\delta \in (0,1)$ to approximate $z = \bra{\psi}U \ket{\psi}$ as $\mathrm{Re}{\bra{\psi}U \ket{\psi}}$ and $\mathrm{Im}{\bra{\psi}U \ket{\psi}}$, each within an additive error of $\xi/\sqrt{2}$ with probability $1-\delta/2$. Note that, to obtain the imaginary part, we have to replace the final $X$ basis measurement of the circuit with a $Y$ basis measurement. By the union bound, we obtain an estimate $\tilde{z}$ with $|\tilde{z} - z| \leq \xi$ with probability $1-\delta$. 
	
	We now generalize the Hadamard Test to expectation values of block-encoded matrices as follows.
	
	\begin{lemma}[Hadamard test for estimating the expectation value of block-encoded matrices]\label{lem:blockHadamard}
		Suppose that $U$ is a block-encoding of $A\in\C^{d\times d}$. Then, performing the Hadamard test on a quantum state $\rho\in\C^{d\times d}$ with a controlled-$U$ and a subsequent $X$ or $Y$ basis measurement yields outcome $0$ with probability $\mathrm{Re}(\mathrm{Tr}[{\rho (I+A)/2}])$ and $\mathrm{Im}(\mathrm{Tr}[{\rho (I+A)/2}])$, respectively.
	\end{lemma}
	\begin{proof}
		
		The probability of getting outcome $0$ with respect to an $X$ basis measurement is
		$$\nrm{(\bra{+}\otimes I)(\ket{0}\ket{\psi}+\ket{1}U\ket{\psi})/\sqrt{2}}^2
		=\frac{1}{4}\nrm{\ket{\psi}+U\ket{\psi}}^2=
		\frac12 + \frac12\mathrm{Re}(\braketbra{\psi}{U}{\psi})
		=\frac12 + \frac12\mathrm{Re}(\tr{\ketbra{\psi}{\psi}U}).
		$$
		
		By linearity we see that for a mixed input state $\rho$ the probability of getting outcome $0$ is
		$$\frac12 + \frac12\mathrm{Re}(\tr{\rho U})=\mathrm{Re}(\tr{\rho (I+U)/2}).
		$$
		
		Now suppose that $A=(\bra{0}\otimes I)U(\ket{0}\otimes I)$, and the input state is $\ketbra{0}{0}\otimes\rho$, then the probability of getting outcome $0$ is
		\begin{equation*}
		\mathrm{Re}(\tr{(\ketbra{0}{0}\otimes\rho) (I+U)/2})=\mathrm{Re}(\tr{(\ketbra{0}{0}\!\otimes\! I)^{\!2}(\ketbra{0}{0}\otimes\rho) (I+U)/2})=\mathrm{Re}(\tr{\rho (I+A)/2}).
		\end{equation*}	
	We remark that the imaginary part can be analogously obtained via a final $Y$ basis measurement.
	\end{proof}

	Let us also describe and alternative method for estimating $\tr{\rho \cdot A^\dagger A }$ using a block-encoding of $A$.
	Suppose that $U$ is a block-encoding of $A$, and we apply $U$ to a quantum state $\rho$ and then measure the defining auxiliary qubits of the block-encoding. The probability of finding them in the $\ket{\bar{0}}$ state is 
	\begin{align*}
	\tr{U (\ketbra{\tilde{0}}{\tilde{0}}\otimes \rho) U^\dagger (\ketbra{\bar{0}}{\bar{0}}\otimes I)}
	=\Tr\Big[\underset{A}{\underbrace{(\bra{\bar{0}}\otimes I) U (\ket{\tilde{0}}\otimes I)}} \cdot \rho \cdot \underset{A^\dagger}{\underbrace{(\bra{\tilde{0}}\otimes I) U^\dagger (\ket{\bar{0}}\otimes I)}}\Big]
	=\tr{\rho \cdot A^\dagger A }.
	\end{align*}
	Note that $A$ can be a rectangular matrix, which is the case if $\bar{0}$ and $\tilde{0}$ have an unequal number of qubits.
	\anote{In the current version I don't think we ever use the above.}

	\subsection{Quantum Singular Value Transformation}\label{subsec:blockMeas}
	Let $f\colon \mathbb{R}\mapsto \mathbb{C}$ be an odd function, i.e., $f(-x)=-f(x)$. For a matrix $A$ with singular value decomposition $A=U \Sigma V^\dagger$ let us denote by $\svt{f}{A}$ the singular value transform of $A$ defined as $\svt{f}{A}:=U f(\Sigma) V^\dagger$. We will heavily use the following fundamental result about quantum implementations of the singular value transform:
	
	\begin{theorem}[Quantum Singular Value Transformation {\cite[Corollary 18 \& Lemma 19]{gilyen2018QSingValTransfArXiv}}]\label{thm:qsvt}
		If $p\in\mathbb{R}[x]$ is an odd polynomial such that $\nrm{p(x)}_{[-1,1]}\leq 1$ and $W$ is a block-encoding of $A$ with $a$ ancillas, then we can implement a block-encoding of $\svt{f}{A}$ with $(\deg(p)+1)/2$ uses of $U$, $(\deg(p)-1)/2$ uses of $U^\dagger$, and $\bigO{a\deg(p)}$ other two-qubit gates.
	\end{theorem}	
	
	In case we need to approximate the singular value transform by using an approximate block-encoding we have the following robustness result from~\cite{gilyen2018QSingValTransf}:
	
	\begin{lemma}[Robustness of singular value transformation {\cite[Lemma 23]{gilyen2018QSingValTransfArXiv}}]\label{lem:PolyNormDiff2}
		If $p\in\C[x]$ is an odd degree-$d$ polynomial such that $\nrm{p}_{[-1,1]}\leq 1$,
		moreover $A,\tilde{A}\in\C^{\tilde{n}\times n}$ are matrices of operator norm at most $1$, such that 
		\begin{equation}\label{eq:nicePropagationBound}
		\nrm{A-\tilde{A}}+\nrm{\frac{A+\tilde{A}}{2}}^2\leq 1,
		\end{equation}
		then we have that 
		$$\nrm{\svt{p}{A}-\svt{p}{\tilde{A}}}\leq d\sqrt{\frac{2}{1-\nrm{\frac{A+\tilde{A}}{2}}^2}}\nrm{A-\tilde{A}}.$$
	\end{lemma}
	The above lemma is proven in~\cite{gilyen2018QSingValTransfArXiv} by showing that there are block-encodings $U, \tilde{U}\in\C^{4(\tilde{n}+n)\times 4(\tilde{n}+n)}$ of $A$ and $\tilde{A}$ respectively so that
	\begin{equation}\label{eq:nicePropagationU}
	\nrm{U-\tilde{U}}\leq \sqrt{\frac{2}{1-\nrm{\frac{A+\tilde{A}}{2}}^2}}\nrm{A-\tilde{A}}.
	\end{equation}	
	Using singular value decomposition one can show that if $U,W\in\C^{d\times d}$ are block-encodings of $A\in\C^{\tilde{n}\times n}$, then there are unitaries $V\in\C^{d-n\times d-n}$, $\tilde{V}\in\C^{d-\tilde{n}\times d-\tilde{n}}$ such that $(I_{\tilde{n}}\oplus \tilde{V})U(I_n\oplus V^\dagger)=W$.\footnote{In fact this is implicit in the proof of the main results of \cite{gilyen2018QSingValTransfArXiv}. \anote{Explain later.}} 
	\begin{corollary}\label{cor:BlockErrorToUnitarError}
		For every block-encoding $U\in\mathbb{C}^{d\times d}$ of $A\in\C^{\tilde{n}\times n}$ with $d\geq 4(\tilde{n}+n)$ and for any $\tilde{A}\in\C^{\tilde{n}\times n}$ satisfying \Cref{eq:nicePropagationBound} we have a block-encoding $\tilde{U}\in\mathbb{C}^{d\times d}$ of $\tilde{A}$ satisfying \Cref{eq:nicePropagationU}.
	\end{corollary}
	
	\begin{corollary}[Error propagation in block-encodings]\label{cor:BlockErrorPropagation}
		Let $d\geq 4(\tilde{n}+n)$; let $W(U)$ be a quantum circuit that uses $U$ and $U^\dagger$ a total of $T$-times. If $W(U)$ implements a block-encoding of $B$ for every block-encoding $U^{d\times d}$ of $A\in\C^{\tilde{n}\times n}$, then $W(\tilde{U})$ implements a block-encoding of $\tilde{B}$ such that if $\tilde{U}$ is a block-encoding of $\tilde{A}\in\C^{\tilde{n}\times n}$ satisfying \Cref{eq:nicePropagationBound} then 
		\begin{equation}\label{eq:nicePropagationBoundGeneral}
			\nrm{B-\tilde{B}}\leq T\sqrt{\frac{2}{1-\nrm{\frac{A+\tilde{A}}{2}}^2}}\nrm{A-\tilde{A}}.
		\end{equation}
	\end{corollary}
	\begin{proof}
		Take some $A\in\C^{\tilde{n}\times n}$, $d\geq 4(\tilde{n}+n)$, and a block-encoding $\tilde{U}$ of some $\tilde{A}\in\C^{\tilde{n}\times n}$ satisfying \Cref{eq:nicePropagationBound}. By \Cref{cor:BlockErrorToUnitarError} there exists some $U$ block-encoding of $A$ such that \Cref{eq:nicePropagationU} holds. Then $W(U)$ is a block-encoding of $B$, and $\nrm{W(U)-W(\tilde{U})}\leq T\nrm{U-\tilde{U}}\leq T\sqrt{\frac{2}{1-\nrm{\frac{A+\tilde{A}}{2}}^2}}\nrm{A-\tilde{A}}$, and therefore \Cref{eq:nicePropagationBoundGeneral} holds as well.
	\end{proof}
	
	\subsection{Low-degree polynomial approximations}
	
	As \Cref{thm:qsvt} suggests in order to optimize our algorithm we will need low-degree polynomial approximations of various functions. We list some polynomial approximation results that we need in our complexity analysis.
	
	Later we will use a low-degree polynomial $p$ approximating the sign function $\mathrm{sgn}\mapsto \{-1,0,1\}$. To construct that we invoke a result of Low and Chuang~\cite[Corollary 6]{low2017HamSimUnifAmp} about constructive polynomial approximations of the sign function -- the error of the optimal approximation, studied by Eremenko and Yuditskii~\cite{eremenko2006uniformApxSgn}, achieves similar scaling but is non-constructive.
	
	\begin{lemma}[Polynomial approximations of the sign function]\label{lemma:signApx}
		For all $\delta>0$ , $\eps\in(0,1/2)$ there exists an efficiently computable odd polynomial $p\in\mathbb{R}[x]$ of degree $n=\bigO{\frac{\log(1/\eps)}{\delta}}$, such that 
		\begin{itemize}
			\item for all $x\in[-2,2]\colon |p(x)|\leq 1$, and
			\item for all $x\in[-2,2]\setminus(-\delta,\delta)\colon |p(x)-\mathrm{sgn}(x)|\leq \eps$.
		\end{itemize}
	\end{lemma}	
	
	\begin{corollary}[Polynomial approximations of rectangle functions {\cite[Corollary~16]{gilyen2018QSingValTransf}}]\label{lemma:polyRect}\hfill\linebreak
		Let $\delta,\eps$ $\in(0,\frac12)$ and $t\in[-1,1]$. There exists an even polynomial $P'\in\R[x]$ of degree $\bigO{\log(\frac1{\eps})/\delta}$, such that $|P'(x)|\leq 1$ for all $x\in[-1,1]$, and
		\begin{equation}\label{eq:polyRect}
		\left\{\begin{array}{rcl} P'(x)\in & [0,\eps] &\text{for all }x\in[-1,-t-\delta]\cup[t+\delta,1]\text{, and}\\
		P'(x)\in & [1-\eps,1] &\text{for all }x\in[-t+\delta,t-\delta].\end{array}\right.
		\end{equation}
	\end{corollary}		
	
	\begin{lemma}[Polynomial approximations of negative power functions {\cite[Cor.\ 3.4.13]{gilyen2018QSingValTransfThesis}}]\label{cor:negatiwePower}\hfill\linebreak
		Let $\delta,\eps\in(0,\frac{1}{2}]$, $c>0$ and let $f(x):=\frac{\delta^c}{2}x^{-c}$, then there exist even/odd polynomials $P,P'\in\R[x]$, such that $\nrm{P-f}_{[\delta,1]}\leq\eps$, $\nrm{P}_{[-1,1]}\leq1$ and similarly $\nrm{P'-f}_{[\delta,1]}\leq\eps$, $\nrm{P'}_{[-1,1]}\leq1$, moreover the degrees of the polynomials are $\bigO{\frac{\max[1,c]}{\delta}\log\left(\frac{1}{\eps}\right)}$.
	\end{lemma}

	\begin{lemma}[Polynomial approximations of positive power functions {\cite[Cor.\ 3.4.14]{gilyen2018QSingValTransfThesis}}]\label{cor:positivePower}\hfill\linebreak
		Let $\delta,\eps\in(0,\frac{1}{2}]$, $c\in (0,1]$ and let $f(x):=\frac{1}{2}x^{c}$, then there exist even/odd polynomials $P,P'\in\R[x]$, such that $\nrm{P-f}_{[\delta,1]}\leq\eps$, $\nrm{P}_{[-1,1]}\leq1$ and similarly $\nrm{P'-f}_{[\delta,1]}\leq\eps$, $\nrm{P'}_{[-1,1]}\leq1$, moreover the degree of the polynomials are $\bigO{\frac{1}{\delta}\log\left(\frac{1}{\eps}\right)}$.
	\end{lemma}	

	\begin{corollary}[Polynomial approximations of positive power functions]\label{cor:posPowerMajorated}
		Let $\delta,\eps\in(0,\frac{1}{2}]$, $c\in (0,1]$ and let $f(x):=\frac{1}{2}x^{c}$, then there exist even/odd polynomials $P,P'\in\R[x]$, such that $\nrm{P-f}_{[\delta,1]}\leq\eps$, $\nrm{P}_{[-1,1]}\leq1$, $\forall x\in [-1,1]\colon|P(x)|\leq f(|x|)+\eps$, and similarly $\nrm{P'-f}_{[\delta,1]}\leq\eps$, $\nrm{P'}_{[-1,1]}\leq1$, $\forall x\in [-1,1]\colon|P'(x)|\leq f(|x|)+\eps$, moreover the degrees of the polynomials are $\bigO{\frac{\max[1,c]}{\delta}\log\left(\frac{1}{\eps}\right)}$.
	\end{corollary}
	\begin{proof}
		Choose $\delta':=\delta/2$ and $\eps':=\eps/2$ in \Cref{cor:positivePower}. Then take a polynomial given by \Cref{lemma:polyRect} for parameters $\delta':=\delta/4$, $t:=\delta/4$ and $\eps':=\eps/2$. The product of the corresponding polynomial satisfies the desired properties.
	\end{proof}

	\subsection{Density matrix exponentiation and block-encodings}

	In case we do not have purified access to the density operators, but can only get independent copies we do not have a direct analog of \Cref{lem:blochDensity}. Instead we will rely on the technique of density matrix exponentiation~\cite{lloyd2013QPrincipalCompAnal}. For this we invoke the following form of the result: \anote{Define diamond norm}
	\begin{theorem}[Density matrix exponentiation {\cite[Theorem 5 \& Theorem 20]{kimmel2016hamiltonian}}]\label{thm:densityMatExp}
		For an unkown quantum state $\rho\in \mathbb{C}^{2^q\times 2^q}$ the sample complexity of implementing the controlled-$e^{it\rho}$ unitary to diamond-norm error $\delta$ is $\Theta(\frac{t^2}{\delta})$ (for the lower-bound one needs to assume $\delta\leq\frac16\min(1,\frac{|t|}{\pi})$). Moreover, the implementation uses $\bigO{q\cdot \frac{t^2}{\delta}}$ two-qubit quantum gates.
	\end{theorem}
	
	As pointed out in~\cite{gilyen2020QAlgForPetzRecovery}, using the results of~\cite{gilyen2018QSingValTransf}, in particular the result about taking the logarithm of a unitary \cite[Corollary 71]{gilyen2018QSingValTransfArXiv} one can implement a block-encoding of $\rho$ also in the sampling access model. To see this we first recall the following result:
	
	\begin{lemma}[Implementing the logarithm of unitaries {\cite[Corollary 71]{gilyen2018QSingValTransfArXiv}}]\label{lem:BlockLogarithm}
		Suppose that $U=e^{iH}$, where $H$ is a Hamiltonian of norm at most $\frac{1}{2}$. Let $\eps\in(0,\frac12]$, then we can implement a $(\frac{2}{\pi},2,\eps)$-block-encoding of $H$ with $\bigO{\log\left(\frac{1}{\eps}\right)}$ uses of controlled-$U$ and its inverse, using $\bigO{\log\left(\frac{1}{\eps}\right)}$ two-qubit gates and using a single ancilla qubit.
	\end{lemma}	
	
	Applying this result to the controlled-$e^{\pm i\frac{\rho}{2}}$ evolution given by \Cref{thm:densityMatExp} we get the following corollary of \Cref{thm:densityMatExp}: 
	
	\begin{corollary}[Sampling to block-encoding]\label{cor:SamplesToBlock}
		For an unknown quantum state $\rho\in \mathbb{C}^{2^q\times 2^q}$ we can implement a quantum operation (quantum channel) that is $\delta$-close in the diamond-norm to an $(\frac{4}{\pi},2,\eps)$-block-encoding of $\rho$ by using $\bigO{ \frac{\log^2(1/\eps)}{\delta}}$ samples of $\rho$ and $\bigO{q\cdot \frac{\log^2(1/\eps)}{\delta}}$ two-qubit quantum gates. In particular we can implement a quantum operation (quantum channel) that is $\delta$-close in the diamond-norm to an $(\frac{4}{\pi},3,0)$-block-encoding of $\rho$ by using $\bigO{ \frac{\log^2(1/\delta)}{\delta}}$ samples of $\rho$ and $\bigO{q\cdot \frac{\log^2(1/\delta)}{\delta}}$ two-qubit quantum gates.
	\end{corollary}
	\begin{proof}
		Due to \Cref{lem:BlockLogarithm} we can implement an $(\frac{4}{\pi},2,\eps)$-block-encoding of $\rho$ with $\bigO{\log(1/\eps)}$ uses of controlled-$e^{\pm i\frac{\rho}{2}}$ unitary. Due to \Cref{thm:densityMatExp} we can approximate each application of the controlled-$e^{\pm i\frac{\rho}{2}}$ unitary to diamond error $\bigO{\delta/\log(1/\eps)}$ using $\bigO{ \frac{\log(1/\eps)}{\delta}}$ samples. This amounts to an overall number of $\bigO{ \frac{\log^2(1/\eps)}{\delta}}$ samples of $\rho$. By the \Cref{lem:BlockLogarithm} and \Cref{thm:densityMatExp} the gate complexity of this implementation is $\bigO{q\cdot \frac{\log^2(1/\eps)}{\delta}}$.
		
		On the other hand a $(\frac{4}{\pi},2,\eps)$-block-encoding also trivially gives rise to a $(\frac{4}{\pi},3,\eps)$-block-encoding by simply adding an extra qubit with the identity operation on it.
		Furthermore, due to \Cref{cor:BlockErrorPropagation} a unitary $(\frac{4}{\pi},3,\eps)$-block-encoding of $\rho$ is $\bigO{\eps}$-close to a unitary $(\frac{4}{\pi},3,0)$-block-encoding of $\rho$, and by definition these unitaries are also $\bigO{\eps}$-close in the diamond norm. So choosing $\delta\leftarrow\frac{\delta'}{2}$ and $\eps\leftarrow\bigO{\delta'}$ the above construction gives a $\delta'$-precise implementation of a $(\frac{4}{\pi},3,0)$-block-encoding of $\rho$ in the diamond norm.
	\end{proof}

	\subsection{Truncated fidelity bounds}\label{subsec:trunFid}
	In general it is difficult to work with density operators that have a large number of tiny eigenvalues that all together represent a significant contribution to the trace. On the other hand, if we filter out small eigenvalues then the problem becomes tractable. Since in general we can only apply soft versions of filtering we need to understand how big is the inaccuracy introduced by such soft truncation. Therefore we devise some slight generalizations of the truncation bounds from \cite[Section 2]{cerezo2019VariaQuantFidEst}.
	
	\begin{lemma}[Monotonicity of fidelity]\label{lem:FidMonotone}
		Let $0 \preceq A \preceq B$ such that $A, B$, and $\rho$ commute with each other, then 
		\begin{align*}
		F(A \rho A,\sigma)
		\leq F(B \rho B,\sigma).
		\end{align*}
	\end{lemma}
	\begin{proof}
		Since $A \rho A\preceq B\rho B$, and the $\sqrt{\cdot}$ function is operator monotone~\cite[Chapter 4.1]{hiai2014IntroMatrixAnalAndApp} we have
		\begin{equation*}
		F(A\rho A,\sigma)
		=\tr{\sqrt{\sqrt{\sigma}A \rho A\sqrt{\sigma}}}
		\leq\tr{\sqrt{\sqrt{\sigma}B\rho B\sqrt{\sigma}}}
		=F(B \rho B,\sigma).\qedhere
		\end{equation*}		
	\end{proof}	
	\anote{In \cite{cerezo2019VariaQuantFidEst} there is a mistake in Appendix 1, Equation 4: the first inequality does not follow from \cite[Theorem 9.7]{nielsen2002QCQI}, but it is corrected in the second inequality nevertheless.}

	\begin{lemma}[Hard truncation bounds]\label{lem:HardProjTrunc}
		Let $\rho,\sigma\succeq 0$ and $\Pi$ be an orthogonal projector that commutes with $\rho$, then 
		\begin{align}\label{def:FidelityBoundsProj}
		F(\rho,\sigma)
		\leq F(\Pi\rho \Pi,\sigma)+\sqrt{\tr{(I-\Pi)\rho(I-\Pi)}}\sqrt{\tr{(I-\Pi)\sigma(I-\Pi)}}.
		\end{align}
	\end{lemma}
	\begin{proof}
		\begin{align*}
			F(\rho,\sigma)
			=\nrm{\sqrt{\rho}\sqrt{\sigma}}_1
			&\leq \nrm{\sqrt{\rho}\Pi\sqrt{\sigma}}_1 + \nrm{\sqrt{\rho}(I-\Pi)\sqrt{\sigma}}_1 \tag{by the triangle inequality}\\	
			&=\nrm{\sqrt{\Pi\rho \Pi}\sqrt{\sigma}}_1 + \nrm{\sqrt{\rho}(I-\Pi)^2\sqrt{\sigma}}_1 \tag{$\Pi \rho = \rho\Pi$ and $(I-\Pi)=(I-\Pi)^2$}\\
			&\leq\nrm{\sqrt{\Pi\rho \Pi}\sqrt{\sigma}}_1 + \|\sqrt{\rho}(I-\Pi)\|_2\cdot\|\sqrt{\sigma}(I-\Pi)\|_2 \tag{by Hölder's inequaltiy}\\							
			&= F(\Pi\rho \Pi,\sigma)+\sqrt{\tr{(I-\Pi)\rho(I-\Pi)}}\cdot\sqrt{\tr{(I-\Pi)\sigma(I-\Pi)}}. \qedhere	
		\end{align*}	
	\end{proof}	
	\begin{corollary}[Soft truncation bounds]\label{cor:SoftBounds}
		Let $\rho,\sigma$ be quantum states and $\thPi{\rho}{I}$ be the orthogonal projector to the subspace spanned by eigenvectors of $\rho$ with eigenvalues in $I$. Let $0\leq \alpha \leq \beta$ and $f\colon \R \mapsto [0,1]$ be such that for all $x< \alpha$ we have $f(x)=0$ and for all $x\geq \beta$ we have $f(x)=1$, then
		\begin{equation}\label{eq:FidelitySoftGuards}
		F(\thPi{\rho}{[\beta,1)} \rho \thPi{\rho}{[\beta,1)},\sigma)
		\leq F(f(\rho) \cdot \rho,\sigma)
		\leq F(\thPi{\rho}{[\alpha,1)} \rho \thPi{\rho}{[\alpha,1)},\sigma)
		\leq F(\rho,\sigma),
		\end{equation}
		and 
		\begin{equation}\label{eq:FidelityCloseness}
		F(\rho,\sigma)-F(f(\rho) \cdot \rho,\sigma)\leq \sqrt{\tr{\thPi{\rho}{[0,\beta)} \rho \thPi{\rho}{[0,\beta)}}}\sqrt{\tr{\thPi{\rho}{[0,\beta)} \sigma \thPi{\rho}{[0,\beta)}}}.
		\end{equation}
	\end{corollary}
	\begin{proof}
		Let $P:=\sqrt{f(\rho)}$, then $\thPi{\rho}{[\beta,1)}\preceq P \preceq \thPi{\rho}{[\alpha,1)}\preceq I$ so \Cref{eq:FidelitySoftGuards} follows from \Cref{lem:FidMonotone} and \Cref{eq:FidelityCloseness} follows from \Cref{eq:FidelitySoftGuards} and \Cref{lem:HardProjTrunc}.
	\end{proof}	
	Replacing $\rho$ in \Cref{lem:HardProjTrunc} by $\Pi_{\alpha}^{\phantom{\rho}} \rho \Pi_{\alpha}^{\phantom{\rho}}$ (letting $\Pi_\theta^{\phantom{\rho}}:=\Pi_{[\theta,1)}^{\phantom{\rho}}$, $\alpha\leq\beta$) we also get the following bound:
	\begin{align}
	F(\Pi_\alpha^{\phantom{\rho}} \rho \Pi_\alpha^{\phantom{\rho}},\sigma) - F(\Pi_\beta^{\phantom{\rho}} \rho \Pi_\beta^{\phantom{\rho}},\sigma)
	&= F(\Pi_\alpha^{\phantom{\rho}} \rho \Pi_\alpha^{\phantom{\rho}},\Pi_\alpha^{\phantom{\rho}} \sigma \Pi_\alpha^{\phantom{\rho}}) - F(\Pi_\alpha^{\phantom{\rho}}\Pi_\beta^{\phantom{\rho}} \rho \Pi_\beta^{\phantom{\rho}}\Pi_\alpha^{\phantom{\rho}},\Pi_\alpha^{\phantom{\rho}} \sigma \Pi_\alpha^{\phantom{\rho}}) \nonumber\\
	&\leq \sqrt{\tr{\thPi{\rho}{[\alpha,\beta)} \rho \thPi{\rho}{[\alpha,\beta)}}}\sqrt{\tr{\thPi{\rho}{[\alpha,\beta)} \sigma \thPi{\rho}{[\alpha,\beta)}}}, \label{eq:FidelityFineGuard}
	\end{align}
	where in the equality we used that for any $\rho\succeq 0$ and $\Pi$ orthogonal projector, we have
	$$
	F(\Pi \rho \Pi, \sigma) 
	= \tr{\sqrt{\sqrt{\Pi \rho \Pi} \sigma \sqrt{\Pi \rho \Pi}}}
	= \tr{\sqrt{\sqrt{\Pi \rho \Pi}\Pi \sigma \Pi\sqrt{\Pi \rho \Pi}}}
	= 	F(\Pi \rho \Pi,\Pi \sigma \Pi).
	$$

	\section{Fidelity estimation with block-encoding based algorithms}\label{sec:block-encoding algorithm}

Let us now sketch the main idea behind our block-encoding algorithm that builds on the Hadamard test.
	Suppose that $\sqrt{\rho}\sqrt{\sigma}$ has singular value decomposition $U \Sigma V^\dagger$, then 
	\begin{align}\label{eq:exactFidelity}
	F(\rho,\sigma)=\nrm{\sqrt{\rho}\sqrt{\sigma}}_1
	=\tr{U^\dagger \sqrt{\rho}\sqrt{\sigma}V }
	=\tr{\sqrt{\rho}\sqrt{\sigma}V U^\dagger}
	=\tr{\rho (\rho^{-\frac{1}{2}}\sqrt{\sigma}V U^\dagger)}.
	\end{align}
	
	By \Cref{lem:blockHadamard} it suffices to implement a (subnormalized) block-encoding of $(\rho^{-\frac{1}{2}}\sqrt{\sigma}VU^\dagger)$ in order to use the Hadamard test for computing an estimate of the fidelity.
	The main issue with this approach is that $\rho^{-\frac{1}{2}}$ can, in general, have arbitrarily large eigenvalues.
	
	In order to deal with singularities arising from small eigenvalues we modify the task as follows. Let $\rho_I:=\thPi{\rho}{I} \rho \thPi{\rho}{I}$ denote the subnormalized density matrix we get after throwing away eigenvalues outside the interval $I$. For some $\delta,\eps\in[0,\frac12]$ we wish to provide an estimate $f$ of $F(\rho_{[\theta,1]},\sigma)$ to precision
	\begin{equation}
		\sqrt{\tr{\rho_{[(1-\delta)\theta,\theta)}}}\sqrt{\tr{\thPi{\rho}{[(1-\delta)\theta,\theta)} \sigma \thPi{\rho}{[(1-\delta)\theta,\theta)}}}+\eps,
	\end{equation}
	in turn providing an estimate of $F(\rho,\sigma)$ with precision
	\begin{equation}\label{eq:truncationError}
		\sqrt{\tr{ \rho_{[0,\theta)}}}\sqrt{\tr{\thPi{\rho}{[0,\theta)} \sigma \thPi{\rho}{[0,\theta)}}}+\eps.
	\end{equation}
	For this we shall use some soft threshold function $t\colon\R\mapsto [0,1]$ such that for all $x< (1-\delta)\theta$ we have $t(x)=0$ and for all $x\geq \theta$ we have $t(x)=1$. By \Cref{eq:FidelityFineGuard} and \Cref{cor:SoftBounds} we have that $f:=F(t^2(\rho)\rho,\sigma)$ satisfies both the above requirements with $\eps=0$, so it suffices to compute $F(t^2(\rho)\rho,\sigma)$ with $\eps$-precision.
	
	In the following we analyze the propagation of errors and the complexity of the implementation.
	
	
	\subsection{Polynomial approximations and error bounds}\label{subsec:BlockFidelityError}
	
	In order to make the procedure more efficient we can approximate $\sqrt{\sigma}$. Let $\frac{1}{2}s\colon\R\mapsto [0,1]$ be a polynomial function provided by \Cref{cor:posPowerMajorated} with parameters $\delta'\leftarrow \frac{\eps^2}{160\rk{t(\rho)}}$ and $\eps'\leftarrow \frac{\eps}{20\sqrt{\rk{t(\rho)}}}$, then $\nrm{s}_{[-1,1]}\leq 1 + \eps'\leq 1 + \eps\leq 2$ and
	\begin{align}
		\left|\nrm{t(\rho)\sqrt{\rho}\sqrt{\sigma}}_1- \nrm{t(\rho)\sqrt{\rho}s(\sigma)}_1\right| 
		&\leq \nrm{t(\rho)\sqrt{\rho}(\sqrt{\sigma}-s(\sigma))}_1 \tag{by the triangle inequality}\\
		&\leq \nrm{t(\rho)}_2\nrm{\sqrt{\rho}}_2\nrm{\sqrt{\sigma}-s(\sigma)} \tag{by Hölder's inequality}\\
		&\leq \sqrt{\rk{t(\rho)}}\cdot 1 \cdot \frac{\eps}{5\sqrt{\rk{t(\rho)}}} = \frac{\eps}{5}. \label{eq:FidelityS}
	\end{align}	
	
	Let us also approximate $\sqrt{\rho}$ by a polynomial. We construct a polynomial $\frac{\sqrt{\theta}}{2\sqrt{2}}q$ by \Cref{cor:negatiwePower}, with parameters $\delta'\leftarrow\frac{\theta}{2}$ and $\eps'\leftarrow\frac{\eps\sqrt{\theta}}{20\sqrt{2}}$, then $\nrm{x q(x)}_{[-1,1]}\leq 2$ and
	\begin{align}
		\left|\nrm{t(\rho)\sqrt{\rho}s(\sigma)}_1- \nrm{t(\rho)\rho q(\rho)s(\sigma)}_1 \right|
		&\leq \nrm{t(\rho)\rho\left(\frac{1}{\sqrt{\rho}}-q(\rho)\right)s(\sigma)}_1 \tag{by the triangle inequality}\\
		&\leq \nrm{\rho}_1\nrm{\frac{t(\rho)}{\sqrt{\rho}}-t(\rho)q(\rho)}\nrm{s(\sigma)} \tag{by Hölder's inequality}\\
		&\leq 1\cdot \frac{\eps}{10} \cdot 2 = \frac{\eps}{5}. \label{eq:FidelitySQ}
	\end{align}	
		
	Let $Q:=t(\rho)\rho q(\rho)s(\sigma)$, by definition we have that 
	\begin{equation}\label{eq:1NormtoTrace}
	\nrm{Q}_1  = \tr{Q \cdot SV^{(\mathrm{sgn})}\!(Q^\dagger)}.
	\end{equation}

	Observe that $Q$ has at most $\rk{t(\rho)}$ non-zero singular values. Also by Hölder's inequality we know that $\nrm{t(\rho)\sqrt{\rho}\sqrt{\sigma}}_1\leq \nrm{t(\rho)}\nrm{\sqrt{\rho}}_2\nrm{\sqrt{\sigma}}_2\leq1$, so by \Cref{eq:FidelityS,eq:FidelitySQ} we get $\nrm{Q}_1\leq 1+\frac{2\eps}{5}\leq \frac{6}{5}$.
	Let $\varsigma_i$ denote the singular values of $Q$. Let $p\colon [-1,1]\mapsto [-1,1]$ be an odd function such that for all $x\in [-1,1]$ with $|x|\geq \frac{\eps}{160\rk{t(\rho)}}$ we have $\mathrm{sgn}(x)-p(x)\leq 5\eps/60$, then
	\begin{align}
	\left|\tr{Q \cdot \svt{\mathrm{sgn}}{Q^\dagger}}-\tr{Q \cdot \svt{p}{Q^\dagger/8}}\right|
	&= \left|\sum_{i=1}\varsigma_{i}-\varsigma_{i}p(\varsigma_{i}/8)\right|\nonumber\\
	&\leq \sum_{i=1}\varsigma_{i}|1-p(\varsigma_{i}/8)|\nonumber\\
	&= \kern-1mm\sum_{\varsigma_{i}\geq \frac{\eps}{20\rk{t(\rho)}}}\kern-3mm\varsigma_{i}|1-p(\varsigma_{i}/8)|
	+\kern-1mm\sum_{\varsigma_{i}< \frac{\eps}{20\rk{t(\rho)}}}\kern-3mm\varsigma_{i}|1-p(\varsigma_{i}/8)|\nonumber\\
	&\leq \sum_{\varsigma_{i}\geq \frac{\eps}{20\rk{t(\rho)}}}\varsigma_{i}\frac{5\eps}{60}
	+\sum_{\varsigma_{i}< \frac{\eps}{20\rk{t(\rho)}}}2\varsigma_{i}\nonumber\\
	&\leq \tr{\Sigma}\frac{\eps}{10} +\sum_{0<\varsigma_{i}< \frac{\eps}{20\rk{t(\rho)}}}2\frac{\eps}{20\rk{t(\rho)}}
	\leq \frac\eps5. \label{eq:SVTError}
	\end{align}
	So performing singular value transformation according to $p$ as opposed to $\mathrm{sgn}$ introduces an error in the estimation that is bounded by $\eps / 5$. Moreover, due to \Cref{lemma:signApx} we can find such a polynomial $p$ with degree $\bigO{\frac{\rk{t(\rho)}}{\eps}\log(1/\eps)}$.
	
	Now we find a polynomial approximation of $t(x)$. We choose a polynomial $\tilde{t}(x):=1-P'(x)$ for a rectangle function from \Cref{lemma:polyRect} with parameters $t'=(1-\frac{\delta}{2})\theta$, $\delta'\leftarrow\frac{\delta}{2}\theta$, and $\eps'\leftarrow\min\{\frac{\sqrt{\theta}\eps}{20\sqrt{2}},\frac{5\eps}{60 \deg(p)}\}$. The degree of $\tilde{t}$ is $\bigO{\frac{1}{\delta\theta}\log\frac{\rk{t(\rho)}}{\theta\eps}}$. Note that we did not yet even specify the shape of $t(x)$ for $x\in [(1-\frac{\delta}{2})\theta,\theta]$, so we simply define it to be $\tilde{t}(x)$ -- resulting in $|t(x)-\tilde{t}(x)|\leq \eps'$ for $x\in[0,1]$. Let us define $\tilde{Q}:=\tilde{t}(\rho)\rho q(\rho)s(\sigma)$, then we wish to bound
	\begin{align}
	\left|\tr{Q \cdot \svt{p}{Q^\dagger/8}}-\tr{\tilde{Q} \cdot \svt{p}{\tilde{Q}^\dagger/8}}\right|
	&\leq \left|\tr{Q \cdot \left(\svt{p}{Q^\dagger/8}-\svt{p}{\tilde{Q}^\dagger/8}\right)}\right|\nonumber\\
	&\phantom{=} + \left|\tr{\left(Q-\tilde{Q}\right) \cdot \svt{p}{\tilde{Q}^\dagger/8}}\right|\tag{triangle inequality}\\
	&\leq \nrm{Q \cdot \left(\svt{p}{Q^\dagger/8}-\svt{p}{\tilde{Q}^\dagger/8}\right)}_1\nonumber\\
	&\phantom{=} + \nrm{\left(Q-\tilde{Q}\right) \cdot \svt{p}{\tilde{Q}^\dagger/8}}_1\tag{trace-norm inequality}\\	
	&\leq \underset{\leq\frac{6}{5}}{\underbrace{\nrm{Q}_1}} \cdot \nrm{\svt{p}{Q^\dagger/8}-\svt{p}{\tilde{Q}^\dagger/8}}\tag{Hölder's ineq.}\\
	&\phantom{=} + \nrm{Q-\tilde{Q}}_1 \cdot \underset{\leq1}{\underbrace{\nrm{\svt{p}{\tilde{Q}^\dagger/8}}}}. \nonumber
	\end{align}
	Now we bound both terms individually as follows
	\begin{align}
		\nrm{Q-\tilde{Q}}_1
		&=\nrm{t(\rho)\rho q(\rho)s(\sigma)}_1 - \nrm{\tilde{t}(\rho)\rho q(\rho)s(\sigma)}_1 \tag{by definition}\\
		&\leq \nrm{(t(\rho)-\tilde{t}(\rho))\rho q(\rho)s(\sigma)}_1 \tag{by the triangle inequality}\\
		&\leq \nrm{(t(\rho)-\tilde{t}(\rho))}\nrm{\rho}_1 \nrm{q(\rho)}\nrm{s(\sigma)} \tag{by Hölder's inequality}\\
		&\leq \eps'\cdot 1 \cdot \frac{2\sqrt{2}}{\sqrt{\theta}} \cdot 2 = \frac{\eps}{10}, \nonumber
	\end{align}		
	and by \Cref{lem:PolyNormDiff2} (choosing $\bar{A}:=\bar{Q}/8$, and observing $\nrm{\bar{Q}}\leq 4$ for $\bar{Q}\in\{Q,\tilde{Q}\}$) we have that
	\begin{align}\label{eq:SVTRrobustness}
	\nrm{\svt{p}{Q/8}-\svt{p}{\tilde{Q}/8}}
	&\leq\deg(p)\cdot\frac{1}{4}\cdot\nrm{\tilde{Q}-Q}\nonumber\\
	&=\frac{\deg(p)}{4}\nrm{(t(\rho)-\tilde{t}(\rho))\rho q(\rho)s(\sigma)}\nonumber\\
	&\leq\frac{\deg(p)}{4}\nrm{t(\rho)-\tilde{t}(\rho)}\nrm{\rho q(\rho)}\nrm{s(\sigma)}\nonumber\\
	&\leq\frac{\deg(p)}{4}\cdot \eps' \cdot 4 \leq \frac{5 \eps}{60},\nonumber
	\end{align}
	ultimately resulting in
	\begin{equation}\label{eq:trunPolyError}
	\left|\tr{Q \cdot \svt{p}{Q^\dagger/8}}-\tr{\tilde{Q} \cdot \svt{p}{\tilde{Q}^\dagger/8}}\right| \leq \frac{\eps}{5}.
	\end{equation}
	
	Therefore, combining \Cref{eq:FidelityS,eq:FidelitySQ,eq:1NormtoTrace,eq:SVTError,eq:trunPolyError} we ultimately get
	\begin{equation}\label{eq:polyFidelity}
	\left|\nrm{t(\rho)\sqrt{\rho}\sqrt{\sigma}}_1- \tr{\tilde{Q} \cdot \svt{p}{\tilde{Q}^\dagger/8}}\right| \leq\frac{4\eps}{5}.
	\end{equation}

	\subsection{Complexity analysis -- with purified access}\label{subsec:purifiedFidCompl}
	
	Assuming the purified access model, we can show the following result for our block-encoding algorithm.
	
	\begin{theorem}
Let $\rho,\sigma \in \mathbb{C}^{d \times d}$ be arbitrary density matrices. Suppose also that $\rho$ has the smallest rank of the two states.
Let $\eps \in (0,1)$, $\delta \in (0,1)$ and $\theta \in (0,1)$ be parameters. Then, given purified access to $\rho$ and $\sigma$, our block-encoding algorithm in \Cref{sec:block-encoding algorithm} runs in time
	\begin{equation*}
	\bigOt{\frac{ \rk{\Pi^\rho_{[(1-\delta)\theta,1]}}}{\eps^2\delta\theta^\frac{3}{2}}T_\rho+\frac{\rk{\Pi^\rho_{[(1-\delta)\theta,1]}}^2}{\eps^4\theta^\frac{1}{2}}T_\sigma}
	\end{equation*}
and outputs (with high probability) an estimate $\hat{F}(\rho_\theta,\sigma)$ such that
$$| F(\rho_\theta,\sigma)-\hat{F}(\rho_\theta,\sigma)| \leq \eps,$$
where $\rho_{\theta}$ is a ``soft-thresholded'' version of $\rho$ in which eigenvalues of $\rho$ below $(1-\delta)\theta$ are removed and those above $\theta$ are kept intact, while eigenvalues in $[(1-\delta)\theta, \theta]$ are decreased by some amount.
\end{theorem}
	
\begin{proof}
	By \Cref{lem:blockHadamard} it suffices to implement a ($\frac{\sqrt{\theta}}{4\sqrt{2}}$ subnormalized) block-encoding of 
	$$\tilde{t}(\rho) q(\rho)s(\sigma)\svt{p}{\tilde{t}(\rho)\rho q(\rho)s(\sigma)/8}$$ in order to use the Hadamard test for computing an estimate of $\tr{\rho\tilde{t}(\rho) q(\rho)s(\sigma)\svt{p}{\tilde{t}(\rho)\rho q(\rho)s(\sigma)/8}}$ -- which by \Cref{eq:polyFidelity} is $\frac{4\eps}{5}$-close to the fidelity. Using amplitude estimation we can estimate the success probability of the Hadamard test to precision $\bigO{\frac{\eps\sqrt{\theta}}{40\sqrt{2}}}$, which then results in an $\eps$ precise estimate of 
	\begin{equation}\label{eq:targetTruncatedFidelity}
		F(t^2(\rho)\rho,\sigma) = \nrm{t(\rho)\sqrt{\rho}\sqrt{\sigma}}_1.
	\end{equation}

	In order to implement a block-encoding of $\tilde{Q}/8$ we implement both $\tilde{t}(\rho)\rho q(\rho)/4$ and $s(\sigma)/2$ using QSVT and take their product~\cite{gilyen2018QSingValTransf}. Sine $\deg(s)=\bigO{\frac{\rk{t(\rho)}}{\eps^2}\log \frac{\rk{t(\rho)}}{\eps}}$ and $\deg(\tilde{t}(x)\cdot x \cdot q(x))$ is $\bigO{\frac{1}{\delta\theta}\log\frac{\rk{t(\rho)}}{\theta\eps}}$, the complexity of implementing $\frac{\tilde{t}(\rho)\rho q(\rho)s(\sigma)}{8}$ is $\bigO{\frac{T_\rho}{\delta\theta}\log\frac{\rk{t(\rho)}}{\theta\eps}+\frac{T_\sigma\rk{t(\rho)}}{\eps^2}\log \frac{\rk{t(\rho)}}{\eps}}$. Applying QSVT by the polynomial $p$ uses the above block-encoding a total of $\bigO{\frac{\rk{t(\rho)}}{\eps}\log(1/\eps)}$ times resulting in complexity $\bigOt{\frac{T_\rho \rk{t(\rho)}}{\eps\delta\theta}+\frac{T_\sigma\rk{t(\rho)}^2}{\eps^3}}$ for implementing $\svt{p}{\tilde{t}(\rho)\rho q(\rho)s(\sigma)/8}$. Similarly we can implement a block-encoding pf $\tilde{t}(\rho) q(\rho)\cdot\frac{\sqrt{\theta}}{2\sqrt{2}}$ by QSVT and take its product with $s(\sigma)/2$, then take a product with $\svt{p}{\tilde{t}(\rho)\rho q(\rho)s(\sigma)/8}$, which only gives a lower order contribution to the running time. Before providing the final runtime bound let us note that $\rk{t(\rho)}\leq \rk{\Pi^\rho_{[(1-\delta)\theta,1]}}$. This all together gives the runtime bound 
	\begin{equation*}
	\bigOt{\frac{ \rk{\Pi^\rho_{[(1-\delta)\theta,1]}}}{\eps^2\delta\theta^\frac{3}{2}}T_\rho+\frac{\rk{\Pi^\rho_{[(1-\delta)\theta,1]}}^2}{\eps^4\theta^\frac{1}{2}}T_\sigma}.
	\end{equation*}
	This proves the claim.
	\end{proof}
	If we have an upper bound on the smallest rank of the two states,
	we obtain the following result.
	
	\begin{corollary}
Let $\rho,\sigma \in \mathbb{C}^{d \times d}$ be arbitrary density matrices. Suppose also that $\rho$ has the smallest rank of the two states, where $\rk{\rho}\leq r$.
Let $\eps \in (0,1)$ be a parameter. Then, our block-encoding algorithm in \Cref{sec:block-encoding algorithm} runs in time
	\begin{equation*}
		\bigOt{\frac{r^\frac{5}{2}}{\eps^5}(T_\rho+T_\sigma)}.
	\end{equation*}
and outputs (with high probability) an estimate $\hat{F}(\rho_\theta,\sigma)$ such that
$$| F(\rho,\sigma)-\hat{F}(\rho_\theta,\sigma)| \leq \eps,$$
where $\rho_{\theta}$ is a ``soft-thresholded'' version of $\rho$ with parameters $\theta = \Theta(\frac{\eps^2}{r})$ and $\delta = \frac{1}{2}$.
\end{corollary}
	
\begin{proof}
Assuming that $\rk{\rho}\leq r$, we can use \Cref{eq:truncationError} and set $\theta = \Theta(\frac{\eps^2}{r})$ and $\delta = \frac{1}{2}$ to obtain an $\eps$-precise estimate $\hat{F}(\rho_\theta,\sigma)$ to the fidelity $F(\rho,\sigma)$ in complexity $		\bigOt{\frac{r^\frac{5}{2}}{\eps^5}(T_\rho+T_\sigma)}$.
\end{proof}
	
	\subsection{Complexity analysis -- with sampling access}
	
	Assuming the purified access model, we can show the following result for our block-encoding algorithm.
	
	\begin{theorem}
Let $\rho,\sigma \in \mathbb{C}^{d \times d}$ be arbitrary density matrices. Suppose also that $\rho$ has the smallest rank of the two states.
Let $\eps \in (0,1)$, $\delta \in (0,1)$ and $\theta \in (0,1)$ be parameters. Then, given sampling access to $\rho$ and $\sigma$, our block-encoding algorithm in \Cref{sec:block-encoding algorithm} uses 
	$$\bigOt{\frac{\rk{t(\rho)}^2}{\eps^5\delta^2\theta^{\frac{7}{2}}}} \text{ copies of } \rho \text{ and} \,\, \bigOt{\frac{\rk{t(\rho)}^4}{\eps^9\theta^{\frac{3}{2}}}} \text{ copies of } \sigma,$$
where $t$ is the threshold function in \Cref{sec:block-encoding algorithm}, and outputs (with high probability) $\hat{F}(\rho_\theta,\sigma)$ such that
$$| F(\rho_\theta,\sigma)-\hat{F}(\rho_\theta,\sigma)| \leq \eps,$$
where $\rho_{\theta}$ is a ``soft-thresholded'' version of $\rho$ in which eigenvalues of $\rho$ below $(1-\delta)\theta$ are removed and those above $\theta$ are kept intact, while eigenvalues in $[(1-\delta)\theta, \theta]$ are decreased by some amount.
\end{theorem}

\begin{proof}
	The overall approach is analogous to \Cref{subsec:purifiedFidCompl}. We implement a $\Theta(\sqrt{\theta})$-subnormalized ($\mathcal{O}(\eps\sqrt{\theta})$-approximate) block-encoding of 
	\begin{equation}\label{eq:targetBlockFidelty}
	\tilde{t}(\rho) q(\rho)s(\sigma)\svt{p}{\tilde{t}(\rho)\rho q(\rho)s(\sigma)/8},
	\end{equation}
	and apply the block-Hadamard test on a copy of $\rho$ as described in \Cref{lem:blockHadamard}. Then it suffices to estimate the probability of outcome $0$ to precision $\bigO{\eps\sqrt{\theta}}$. Since in this scenario we only get copies of $\rho$ we cannot implement amplitude estimation (which would require the ability to prepare $\rho$), but need to repeat the Hadamard test a total of $\frac{1}{\eps^2\theta}$ times to get an estimate with such precision.
	
	Another difference from \Cref{subsec:purifiedFidCompl} is that we do not natively have a perfect block-encoding of $\rho$ and $\sigma$. However, using density matrix exponentiation by \Cref{cor:SamplesToBlock} we can get an approximate block-encoding of a $\frac{\pi}{4}$-subnormalized block-encoding of $\rho$ and $\sigma$. The $\frac{\pi}{4}$-subnormalization constant only induces constant factor changes in our analysis in \Cref{subsec:BlockFidelityError}. In particular our earlier analysis in \Cref{subsec:purifiedFidCompl} still shows that a $\Theta(\sqrt{\theta})$-subnormalized block-encoding of \eqref{eq:targetBlockFidelty} can be implemented with $b_\rho:=\bigOt{\frac{\rk{t(\rho)}}{\eps\delta\theta}}$ uses of $U_\rho$ a block-encoding of $\frac{\pi}{4}\rho$ and $b_\sigma:=\bigOt{\frac{\rk{t(\rho)}^2}{\eps^3}}$ uses of $U_\sigma$ a block-encoding of $\frac{\pi}{4}\sigma$.
	We can implement $U_\rho$ to $\bigO{\frac{\eps\sqrt{\theta}}{b_\rho}}$-error in the diamond norm by using $\bigOt{\frac{b_\rho}{\eps\sqrt{\theta}}}$ copies of $\rho$ and similarly implement $U_\sigma$ to $\bigO{\frac{\eps\sqrt{\theta}}{b_\sigma}}$-error in the diamond norm by using $\bigOt{\frac{b_\sigma}{\eps\sqrt{\theta}}}$ copies of $\sigma$ due to \Cref{cor:SamplesToBlock}. This then results in an implementation of a block-encoding of a $\Theta(\sqrt{\theta})$-subnormalized version of \eqref{eq:targetBlockFidelty} up to diamond-norm error $\bigO{\eps\sqrt{\theta}}$ using $\bigOt{\frac{b_\rho^2}{\eps\sqrt{\theta}}}$ copies of $\rho$ and $\bigOt{\frac{b_\sigma^2}{\eps\sqrt{\theta}}}$ copies of $\sigma$. This construction ensures that the probability of getting outcome $0$ in the Hadamard-test is $\bigO{\eps\sqrt{\theta}}$-close to the outcome one would get by using an exact block-encoding of \eqref{eq:targetBlockFidelty}. By appropriately choosing the constants this ensures that repeating the Hadamard-test  $\frac{1}{\eps^2\theta}$ times with high-probability we get an $\eps$-precisie estimate of \eqref{eq:targetTruncatedFidelity}.
	
	This amounts to an algorithm that uses $\bigOt{\frac{b_\rho^2}{\eps^3\theta^{\frac{3}{2}}}}$ and $\bigOt{\frac{b_\sigma^2}{\eps^3\theta^{\frac{3}{2}}}}$ copies of $\rho$ and $\sigma$ respectively, i.e., the ultimate algorithm uses 
	$$\bigOt{\frac{\rk{t(\rho)}^2}{\eps^5\delta^2\theta^{\frac{7}{2}}}} \text{ copies of } \rho$$
	and
	$$\bigOt{\frac{\rk{t(\rho)}^4}{\eps^9\theta^{\frac{3}{2}}}} \text{ copies of } \sigma.$$
	The implementation is also gate efficient, the gate complexity overhead of the implementation is $\bigO{\log(\dim(\rho))}$ as follows from \Cref{cor:SamplesToBlock}. This proves the claim.
	\end{proof}
	
	If we have an upper bound on the smallest rank of the two states, we get the following:
	
    	\begin{corollary}
Let $\rho,\sigma \in \mathbb{C}^{d \times d}$ be arbitrary density matrices. Suppose also that $\rho$ has the smallest rank of the two states, where $\rk{\rho}\leq r$.
Let $\eps \in (0,1)$ be a parameter. Then, given sampling access to $\rho$ and $\sigma$, our block-encoding algorithm in \Cref{sec:block-encoding algorithm} uses $\bigOt{\frac{r^{5.5}}{\eps^{12}}}$ copies of $\rho$ and $\sigma$ and outputs (with high probability) an estimate $\hat{F}(\rho_\theta,\sigma)$ such that
$$| F(\rho,\sigma)-\hat{F}(\rho_\theta,\sigma)| \leq \eps,$$
where $\rho_{\theta}$ is a ``soft-thresholded'' version of $\rho$ with parameters $\theta = \Theta(\frac{\eps^2}{r})$ and $\delta = \frac{1}{2}$.
 \end{corollary}
\begin{proof}
	In case we know that $\rk{\rho}\leq r$, by \Cref{eq:truncationError} we can see that by setting $\theta = \Theta(\frac{\eps^2}{r})$ and $\delta =\frac{1}{2}$ we can obtain an $\eps$-precise estimate with $\bigOt{\frac{r^{5.5}}{\eps^{12}}}$ copies of $\rho$ and $\sigma$.
\end{proof}	
	\anote{
	\subsection{Alternative algorithm?}
	
	$$F(\rho,\sigma)=\tr{\sqrt{\sqrt{\rho}\sigma\sqrt{\rho}}}=\mathrm{rank}(\rho)\tr{\rho_{\mathrm{supp}}\sqrt{\sqrt{\rho}\sigma\sqrt{\rho}}}$$
	
	Implement $\sqrt{\sqrt{\rho}\sigma\sqrt{\rho}}$ thresholded at $\frac{\eps^2}{\mathrm{rank}^2(\rho)}$ with precision $\eps/\mathrm{rank}(\rho)$ with complexity $\bigOt{\frac{\mathrm{rank}^2(\rho)}{\eps^2} T_{\sqrt{\rho}\sigma\sqrt{\rho}}}$
	}

	\section{Fidelity estimation via spectral sampling}\label{sec:spectral-sampling-algorithm}
	
	In this section, we present our second approximation algorithm for the problem of fidelity estimation. Recall that we can write the fidelity between two  states $\rho,\sigma \in \mathbb{C}^{d \times d}$ as the quantity $F(\rho,\sigma) = \Tr{[\sqrt{\Lambda}]}$, where
	$\Lambda(\rho,\sigma)=\sqrt{\rho}\sigma\sqrt{\rho} \in \mathbb{C}^{d \times d}$ has the following non-trivial entries in the eigenbasis of $\rho$:
	\begin{align*}
	 \Lambda_{ij}  = \sqrt{\lambda_i}\sqrt{\lambda_j}  \bra{\psi_i} \sigma \ket{\psi_j}, \quad\quad \forall i,j \in \{1,\dots,\rk{\rho}\}.
	\end{align*}
	Hence, it suffices to directly compute the matrix elements of $\Lambda$ in order to estimate the fidelity $F(\rho,\sigma)$. In \Cref {sec:continouity}, we show a continuity bound for fidelity estimation. This allows us to quantify the approximation error of our estimate $\hat{F}(\rho,\sigma)=\mathrm{Tr}\big[{\sqrt{\hat{\Lambda}_+}}\big]$ for $F(\rho,\sigma) = \Tr{[\sqrt{\Lambda}]}$ in terms of the approximation error between $\Lambda$ and $\hat{\Lambda}$, where $\hat{\Lambda}$ is our initial estimate and $\hat{\Lambda}_+$ is the projection onto the positive semidefinite cone. In particular, we show in \Cref{lem:continuity} that
	$$
|F(\rho,\sigma) - \hat{F}(\rho,\sigma) |  \,\leq\,  \sqrt{2} r \cdot \sqrt{\left\|\Lambda - \hat{\Lambda} \right\|_1},
$$
where $r = \min \{ \rk{\rho},\rk{\sigma}\}$ denotes the smallest rank of the states $\rho$ and $\sigma$.
	To estimate the eigenvalues of $\rho$, we rely on the technique called \emph{quantum spectral sampling} first introduced by Lloyd, Mohseni and Rebentrost~\cite{lloyd2013QPrincipalCompAnal} in the context of \emph{quantum principal component analysis}.

\subsection{Continuity bound for fidelity estimation}\label{sec:continouity}

In this section, we prove an important technical result which allows us to relate the approximation error of a fidelity estimate for $F(\rho,\sigma) = \Tr{[\sqrt{\Lambda}]}$ in terms of the approximation error for the matrix $\Lambda(\rho,\sigma)=\sqrt{\rho}\sigma\sqrt{\rho}$. In other words, if $\Lambda$ and $\hat{\Lambda}$ are close, then the fidelity estimate $\hat{F}(\rho,\sigma)=\mathrm{Tr}\big[{\sqrt{\hat{\Lambda}_+}}\big]$ is close to $F(\rho,\sigma) = \Tr{[\sqrt{\Lambda}]}$, where $\hat{\Lambda}_+$ is the projection onto the positive semidefinite cone.

\begin{theorem}[Continuity bound]\label{lem:continuity} Let $\rho,\sigma \in \mathbb{C}^{d \times d}$ be density matrices and $\Lambda(\rho,\sigma)=\sqrt{\rho}\sigma\sqrt{\rho}$. Let $\hat{\Lambda}\in \mathbb{C}^{d \times d}$ be an arbitrary hermitian matrix with $\mathrm{rk}(\hat \Lambda)\leq \rk{\Lambda}$ and suppose that $\hat{F}(\rho,\sigma)=\mathrm{Tr}\big[{\sqrt{\hat{\Lambda}_+}}\big]$, where $\hat{\Lambda}_+$ is the projection  of $\hat{\Lambda}$ onto the positive semidefinite cone.  Then,
$$
|F(\rho,\sigma) - \hat{F}(\rho,\sigma) | = \big| \mathrm{Tr} \big[{\sqrt{\Lambda}}\big] - \mathrm{Tr}\big[{\sqrt{\hat{\Lambda}_+}}\big] \big| \,\leq\,  \sqrt{2} r \cdot \sqrt{\left\|\Lambda - \hat{\Lambda} \right\|_1},
$$
where we let $r = \min \{ \rk{\rho},\rk{\sigma}\}$ denote the smallest rank of the states $\rho$ and $\sigma$.
\end{theorem}
\begin{proof}
Note that $\Lambda(\rho,\sigma)=\sqrt{\rho}\sigma\sqrt{\rho}$ is positive semidefinite and hermitian with $\rk{\Lambda}\leq r$. 
Let $\spec(\Lambda) = (\lambda_1,\dots,\lambda_r)$ and $\spec(\hat{\Lambda}_+) = (\hat\lambda^+_1,\dots,\hat\lambda^+_r)$ be the spectra of $\Lambda$ and $\hat{\Lambda}_+$, respectively.
Then,
\begin{align*}
\big| \mathrm{Tr} \big[{\sqrt{\Lambda}}\big] - \mathrm{Tr}\big[{\sqrt{\hat{\Lambda}^+}}\big] \big|
&= \big|\sum_{i=1}^r  \sqrt{\lambda_i} - \sum_{i=1}^r \sqrt{\hat\lambda^+_i} \big| & (\text{by definition})\\
&\leq \sum_{i=1}^r \big|\sqrt{\lambda_i}  - \sqrt{\hat\lambda_i^+} \big|  & (\text{triangle inequality})\\
&\leq r \cdot \max_{1 \leq i \leq r} \big|\sqrt{\lambda_i}  - \sqrt{\hat\lambda_i^+} \big| \\
&\leq r \cdot \left\|\sqrt{\Lambda} - \sqrt{\hat{\Lambda}_+} \right\|_2 & (\text{\cite[Problem III.6.13]{bhatia1997MatrixAnalysis}})\\
&\leq r \cdot \sqrt{\left\|\Lambda - \hat{\Lambda}_+ \right\|_1} & (\text{\cite[Lemma 4.1]{cmp/1103842028}})
\end{align*}
Note that $\hat{\Lambda} = \hat{\Lambda}_{+} + \hat{\Lambda}_{-}$.
Thus, by the triangle inequality, we obtain
\begin{align}\label{eq:lambda-plus-bound1}
 \left\|\Lambda - \hat{\Lambda}_+ \right\|_1 = \left\|\Lambda - \hat{\Lambda} + \hat{\Lambda}_{-} \right\|_1   \leq \left\|\Lambda - \hat{\Lambda}\right\|_1 +\left\| \hat{\Lambda}_{-} \right\|_1.
\end{align}
Let us now bound the quantity $\| \hat{\Lambda}_{-} \|_1$.
Recall from Lemma \ref{lem:projection_cone} that the projection $\hat{\Lambda}_+$ satisfies
\begin{align} \label{eq:argmin-PSD}
	\hat{\Lambda}_+ = \underset{X \succeq 0}{\mathrm{argmin}} \left\|  X -\hat{\Lambda} \right\|_1.
\end{align}
In other words, Eq.~\eqref{eq:argmin-PSD} yields the following inequality
\begin{align}\label{eq:lower-bound-with-X}
\left\| X-\hat{\Lambda}  \right\|_1 \geq \left\| \hat{\Lambda}_+ - \hat{\Lambda} \right\|_1 = \left\| \hat{\Lambda}_{-} \right\|_1, \quad\quad \forall X \succeq 0.    
\end{align}
Using the fact that $\Lambda \succeq 0$, Eq.~\eqref{eq:lower-bound-with-X} implies the following upper bound given by
\begin{align}\label{eq:lower-bound-lambda-lambda-hat}
\left\| \hat{\Lambda}_{-} \right\|_1 \leq  \left\| \Lambda - \hat{\Lambda} \right\|_1 .
\end{align}
Putting everything together, we can use \eqref{eq:lambda-plus-bound1} and \eqref{eq:lower-bound-lambda-lambda-hat} to conclude that
\begin{align*}
|F(\rho,\sigma) - \hat{F}(\rho,\sigma) | &=  \big| \mathrm{Tr} \big[{\sqrt{\Lambda}}\big] - \mathrm{Tr}\big[{\sqrt{\hat{\Lambda}_+}}\big] \big|\\
&\leq r \cdot \sqrt{\left\|\Lambda - \hat{\Lambda}_+ \right\|_1}\\
&= \sqrt{2} r  \cdot \sqrt{\left\|\Lambda - \hat{\Lambda} \right\|_1}. \quad\quad\quad
\end{align*}
This proves the claim.
\end{proof}

\subsection{Algorithm}

	Let us now state our second approximation algorithm for fidelity estimation via quantum spectral sampling.
	Recall that, in Section \ref{subsec:swap-test} and Section \ref{subsec:HadTest}, we reviewed the \emph{Swap Test} and \emph{Hadamard Test}, respectively. We will review the \emph{quantum spectral sampling} algorithm in Section \ref{subsec:quantum-spectral-sampling} and the \emph{quantum eigenstate filtering} algorithm in Section \ref{subsec:quantum-eigenstate-filtering}.

Finally, we also remark that the analysis of the algorithm is given in Section \ref{subsec:analysis-spectral-sampling}.
	
	\begin{algorithm}\caption{Fidelity estimation via spectral sampling}\label{algorithm:spectral_sampling_fidelity}
\DontPrintSemicolon
\SetAlgoLined
\KwIn{
Purified access to density operators $\rho,\sigma \in \mathbb{C}^{d \times d}$ via unitaries $U_\rho$ and $U_\sigma$.\\ \textbf{Promise:} $\rho=\sum_{i=1}^{\rk{\rho}} \lambda_i \ket{\psi_i}\bra{\psi_i}$ has well-separated spectrum with gap $\Delta >0$. \\
\textbf{Parameters:} $\eps \in (0,1)$, $\delta \in (0,1)$, $\theta \in (0,1)$.\ \\}
Set $m \leftarrow \lceil \frac{1}{2\theta} \rceil$.\;
\RepTimes{$\left\lceil \frac{16 \cdot m H_m}{\delta \cdot \theta^2}  \right\rceil$}{
    Run the subroutine \texttt{Quantum\_Spectral\_Sampling}$(U_\rho,\gamma,\ell)$ with  $\gamma = \min \{\frac{ \theta^3 \eps}{\sqrt{2}},\frac{\Delta}{2}\}$ and $\ell = \left\lceil \log\big(\frac{16}{\delta \cdot \theta^2}\big) + \log \lceil \frac{16 m H_m}{\delta \cdot \theta^2} \rceil\right\rceil$ to generate the bipartite state
    $$ \sum_{i=1}^{\rk{\rho}} \lambda_i \, \ket{\psi_i}\bra{\psi_i} \otimes \ket{\tilde \lambda_i}\bra{\tilde \lambda_i} $$\\
    Measure the second register in the computational basis to obtain a sample $\tilde{\lambda}_j$. \;
    Discard $\tilde{\lambda}_j$ if it is smaller than $\frac{3\theta}{4}$ or within distance $\frac{\Delta}{2} + \gamma$ of any previously seen sample.
}
Let $(\tilde{\lambda}_1,\dots,\tilde{\lambda}_{r_\theta})$ be the collected eigenvalues in decreasing order.\\
\For{$i=1$ to $r_\theta$}{
\For{$j=i$ to $r_\theta$}{
  \eIf{$i=j$}{
    Run \texttt{Swap\_Test}$(U_i \ket{0^{\log d}},\sigma,\xi,\nu)$ with $\xi=\frac{\eps^2}{8r_\theta^4}$ and $\nu = \frac{\delta}{2r_\theta^2}$ to obtain $\tilde{\sigma}_{ii} \approx \bra{\psi_i}\sigma\ket{\psi_i}$, for the unitary $U_i =\texttt{Quantum\_Eigenstate\_Filtering}(U_\rho,\tilde{\lambda}_i,\frac{\eps^2}{16r_\theta^4})$ in Theorem \ref{thm:quantum-eigenstate-filtering}\;
   Set $\hat{\Lambda}_{ii}^\theta \leftarrow \tilde{\lambda}_i \tilde{\sigma}_{ii}$\;
   }{
   Run \texttt{Hadamard\_Test}$(\ket{0^{\log d}}, U_i^\dag U U_j,\xi,\nu)$ with parameters $\xi = \frac{\eps^2}{8r_\theta^4}$ and $\nu = \frac{\delta}{2r_\theta^2}$, for the unitaries
   $U_j =\texttt{Quantum\_Eigenstate\_Filtering}(U_\rho,\tilde{\lambda}_j,\frac{\eps^2}{16r_\theta^4})$ and $U_i^\dag =\texttt{Quantum\_Eigenstate\_Filtering}(U_\rho,\tilde{\lambda}_i,\frac{\eps^2}{16r_\theta^4})^\dag$ in Theorem \ref{thm:quantum-eigenstate-filtering}, and where $U$ is a block-encoding of $\sigma$ as in Lemma \ref{lem:blochDensity}, to obtain $\tilde{\sigma}_{ij} \approx  \bra{\psi_i} \sigma \ket{\psi_j}$. \;
   Set $\hat{\Lambda}^\theta_{ij} \leftarrow \sqrt{\tilde{\lambda}_i}\sqrt{\tilde{\lambda}_j}\,\tilde{\sigma}_{ij}$\;
   Set $\hat{\Lambda}^\theta_{ji} = \hat{\Lambda}^{\theta^*}_{ij}$
   }
  }
 }
Output the estimate $\hat{F}(\rho_\theta,\sigma) = \mathrm{Tr}\left[\sqrt{\hat{\Lambda}_+^{\theta}}\right]$
\end{algorithm}
	
	\subsection{Quantum spectral sampling}\label{subsec:quantum-spectral-sampling}
	
	The following definition of the infinitesimal SWAP operation allows us to approximately implement the density matrix exponential $U=e^{-2\pi i \rho}$, as shown by Lloyd, Mohseni and Rebentrost~\cite{lloyd2013QPrincipalCompAnal} in the context of \emph{quantum principal component analysis}, and later extended by Prakash \cite{prakash2014QLinAlgAndMLThesis}.

\begin{definition}[Infinitesimal SWAP operation]\label{def:inf_SWAP}

Let $k \in \mathbb{N}$ be a parameter and let $\rho$ be a mixed state. We define the action $U^{(k)} \sigma U^{(k)}{}^\dag := \sigma^{(k)}$ on an input state $\sigma$ implicitly via the following iteration:
\begin{align*}
& \sigma^{(0)} = \sigma\\
& \sigma^{(n+1)} = \mathrm{Tr}_B\big[ e^{- 2 \pi i S / k} \big(\sigma^{(n)}_A \otimes \rho_B \big) e^{2 \pi i S / k}\big], \quad 0 \leq n \leq k-1,
\end{align*}
where $S$ is the SWAP operator.
\end{definition}

Prakash~\cite{prakash2014QLinAlgAndMLThesis} showed that the procedure $U^{(k)} \approx e^{-2\pi i \rho}$ can be implemented in time $O(k T_\rho)$, where $T_\rho$ is the time it takes to prepare the state $\rho$ given purified access $U_\rho$. In particular, the following theorem due Prakash \cite[Theorem 3.2.1]{prakash2014QLinAlgAndMLThesis} states that it is possible to sample from the eigenvalues of a density operator $\rho \in \mathbb{C}^{d \times d}$ in time $\tilde{O}(T_\rho/\gamma^3)$ with high probability.

\begin{theorem}[\cite{prakash2014QLinAlgAndMLThesis}]
Let $\rho \in \mathbb{C}^{d \times d}$ be a density matrix, and let $T_\rho$ denote the time it takes to prepare $\rho$ via purified access to the unitary $U_\rho$. Let $\gamma \in (0,1)$ and $\ell \in \mathbb{N}$ be parameters.
Then, the procedure \texttt{Quantum\_Spectral\_Sampling}$(U_\rho,\gamma,\ell)$ in Algorithm \ref{alg:quantum-spectral-sampling}
runs in time $\tilde{O}(\ell \cdot T_\rho/\gamma^3)$
and produces eigenvalue estimates $\tilde{\lambda}_j$ such that $|\tilde{\lambda}_j - \lambda_j| \leq \gamma $ with probability at least $1 - 2^{-\ell}$.
\end{theorem}


\begin{algorithm}[H]\label{alg:quantum-spectral-sampling}
\DontPrintSemicolon
\SetAlgoLined
\KwIn{Purified access to $\rho \in \mathbb{C}^{d \times d}$ with spectral decomposition $\rho=\sum_{i=1}^{\rk{\rho}} \lambda_i \ket{\psi_i}\bra{\psi_i}$ via $U_\rho$. Parameters $\gamma \in (0,1)$ and $\ell \in \mathbb{N}$. \ \\
\ \\}

\RepTimes{$\ell$}{
    Run \texttt{Quantum\_Phase\_Estimation} on $\rho$ and the simulated unitary $U_k$ in Definition \ref{def:inf_SWAP} with precision $\gamma$ and parameter $k = \lceil\frac{200 \log(d\log(1/\gamma)/\gamma))}{\gamma^2}\rceil$ \;
}
Post-select the most frequently observed estimates $\tilde{\lambda}_i$ to obtain the state
    $$
    \sum_{i=1}^{\rk{\rho}} \lambda_i \, \ket{\psi_i}\bra{\psi_i} \otimes \ket{\tilde \lambda_i}\bra{\tilde \lambda_i}
    $$
 \caption{\texttt{Quantum\_Spectral\_Sampling}}
\end{algorithm}

	\subsection{Quantum eigenstate filtering}\label{subsec:quantum-eigenstate-filtering}
	
	The following theorem is implicit in the work of Lin and Tong~\cite[Theorem 3]{lin2019OptimalQEigenstateFiltering} and states that we can approximately prepare a given eigenstate $\ket{\psi_i}$ of a density operator $\rho=\sum_{i=1}^{r} \lambda_i \ket{\psi_i}\bra{\psi_i}$ up to precision $\eps \in (0,1)$ in time $\poly\big(\Delta^{-1},1/\sqrt{\lambda_i},\log\left(1/\eps\right),T_\rho\big)$, where $r = \rk{\rho}$ is the rank of $\rho$ and where $T_\rho$ is the time it takes to prepare a purification if $\rho$.
	
	\begin{theorem}[\cite{lin2019OptimalQEigenstateFiltering}]\label{thm:quantum-eigenstate-filtering}
	Let $\rho=\sum_{i=1}^{r} \lambda_i \ket{\psi_i}\bra{\psi_i} \in \mathbb{C}^{d \times d}$ be a density operator with rank $r = \rk{\rho}$ and $\Delta$-gapped spectrum $\spec(\rho) = (\lambda_1,\dots,\lambda_r)$, for some $\Delta > 0$. Then, given purified access $U_\rho$ to $\rho$, one can construct a unitary quantum circuit $U_i =\texttt{Quantum\_Eigenstate\_Filtering}(U_\rho,\lambda_i,\eps)$ (and its inverse $U_i^\dag$) with respect to an eigenvalue $\lambda_i$ and parameter $\eps \in (0,1)$ which takes as input $\ket{0^{\log d}}$ and generates an approximate eigenstate $\ket{\tilde{\psi}_i}$ in time $O\left(\frac{\log\left(\frac{1}{\eps}\right) T_\rho}{\Delta \cdot \sqrt{\lambda_i}} \right)$ such that:
	$$
	\| \ket{\tilde{\psi}_i} - \ket{\psi_i} \| \leq \eps.
	$$
	\end{theorem}

	\subsection{Technical lemmas}

\begin{lemma}\label{lem:approximate-inner-products}
Let $\ket{\psi}, \ket{\phi} \in \mathbb{C}^d$ be pure states and $\sigma \in \mathbb{C}^{d \times d}$ a density matrix, and suppose that there exist pure states $\ket{\tilde{\psi}}, \ket{\tilde{\phi}} \in \mathbb{C}^d$
such that
$\|\ket{\tilde{\psi}} - \ket{\psi} \| \leq \eps$ and $\|\ket{\tilde{\phi}} - \ket{\phi} \| \leq \eps$, for $\eps \in (0,1)$. Then,
$$
\big| \bra{\tilde{\psi}}\sigma \ket{\tilde{\phi}}  - \bra{\psi}\sigma \ket{\phi} \big| \leq 2 \eps.
$$
\end{lemma}
\begin{proof}
Using that the states are normalized, we find that
\begin{align*}
\big| \bra{\tilde{\psi}}\sigma \ket{\tilde{\phi}}  - \bra{\psi}\sigma \ket{\phi}  \big| &\leq 
\big| \bra{\tilde{\psi}}\sigma \ket{\tilde{\phi}}  - \bra{\tilde{\psi}} \sigma \ket{\phi} \big| +
\big| \bra{\tilde{\psi}}\sigma \ket{\phi}  - \bra{\psi}\sigma \ket{\phi} \big| \\
&=
\big| \bra{\tilde{\psi}} \sigma
(\ket{\tilde{\phi}}  - \ket{\phi} )\big| +
\big| (\bra{\tilde{\psi}}  - \bra{\psi}) \sigma \ket{\phi} \big| & (\text{linearity})\\
&=
\big| \bra{\tilde{\psi}} \sigma
(\ket{\tilde{\phi}}  - \ket{\phi} )\big| +
\big| \bra{\phi} \sigma (\ket{\tilde{\psi}}  - \ket{\psi})^* \big| & (\text{skew symmetry})\\
&\leq \|\tilde{\psi}\| \cdot \|\sigma\| \cdot \|
\ket{ \tilde{\phi}}  - \ket{\phi} \| +
\|\phi\| \cdot \|\sigma\| \cdot \|\ket{\tilde{\psi}}  - \ket{\psi} \|& (\text{Cauchy-Schwarz})\\
&\leq 2\eps.
\end{align*}
\end{proof}

\begin{lemma}[Estimation errors for diagonal terms]\label{lem:diagonal-terms} Let $\rho,\sigma \in \mathbb{C}^{d \times d}$ be density matrices, where $\rho$ has the spectral decomposition $\rho=\sum_{i=1}^{r} \lambda_i \ket{\psi_i}\bra{\psi_i}$.
Suppose that, for every $i \in [r]$, there exist estimates $ \tilde{\lambda}_i \in (0,1]$ and $\tilde{\sigma}_{ii} \in (0,1]$ as well as parameters $\gamma,\xi,\tau,\nu \in (0,1)$ such that 
\begin{itemize}
    \item $| \tilde{\lambda}_i - \lambda_i| \leq \gamma$ with probability at least $1-\tau$,
    \item $|\tilde{\sigma}_{ii} - \bra{\psi_i}\sigma\ket{\psi_i}| \leq \xi$ with probability at least $1 - \nu$.
\end{itemize}
Then, for $i \in [r]$, it holds with probability at least $1 - \tau - \nu$:
$$
\big| \tilde{\lambda}_i \tilde{\sigma}_{ii} - \lambda_i \bra{\psi_i}\sigma\ket{\psi_i} \big| \leq \gamma + \xi.
$$
\end{lemma}
\begin{proof}
Fix $i \in [r]$. By the union bound we get that with probability at least $1-\tau -\nu$:
\begin{align*}
\big| \tilde{\lambda}_i \tilde{\sigma}_{ii} - \lambda_i \bra{\psi_i}\sigma\ket{\psi_i} \big|
&=| \tilde{\lambda}_i \tilde{\sigma}_{ii} - \tilde{\lambda}_i \bra{\psi_i}\sigma\ket{\psi_i} + \tilde{\lambda}_i \bra{\psi_i}\sigma\ket{\psi_i} -\lambda_i \bra{\psi_i}\sigma\ket{\psi_i} |\\
&\leq \tilde{\lambda}_i \cdot| \tilde{\sigma}_{ii} -\bra{\psi_i}\sigma\ket{\psi_i} | + |\bra{\psi_i}\sigma\ket{\psi_i}| \cdot |\tilde{\lambda}_i  - {\lambda}_i| & (\text{triangle inequality})\\
&\leq \tilde{\lambda}_i \cdot| \tilde{\sigma}_{ii} -\bra{\psi_i}\sigma\ket{\psi_i} | + \|\psi_i\| \cdot \|\sigma\ket{\psi_i}\| \cdot |\tilde{\lambda}_i  - {\lambda}_i| & (\text{Cauchy-Schwarz})\\
&\leq \tilde{\lambda}_i \cdot| \tilde{\sigma}_{ii} -\bra{\psi_i}\sigma\ket{\psi_i} | + \|\psi_i\| \cdot \|\sigma\| \cdot \|\psi_i\| \cdot |\tilde{\lambda}_i  - {\lambda}_i| & (\text{consistency of norm})\\
&\leq \tilde{\lambda}_i \cdot| \tilde{\sigma}_{ii} -\bra{\psi_i}\sigma\ket{\psi_i} | + \|\sigma\|_1 \cdot |\tilde{\lambda}_i  - {\lambda}_i| & (\text{since } \|\sigma\| \leq \|\sigma\|_1 )\\
&\leq \eps + \xi,
\end{align*}
where in the last last we used that $0 \leq \tilde{\lambda}_i \leq 1$, for all $i \in [r]$, and that $\|\sigma\|_1 =1$.
\end{proof}

\begin{lemma}[Estimation errors for off-diagonal terms]\label{lem:off-diagonal}
Let $\rho,\sigma \in \mathbb{C}^{d \times d}$ be density matrices, where $\rho$ has the spectral decomposition $\rho=\sum_{i=1}^{r} \lambda_i \ket{\psi_i}\bra{\psi_i}$.
Let $t \in (0,1)$ and suppose that, for every $i,j\in [r]$ with $i < j$, there exist $ \tilde{\lambda}_i ,\tilde{\sigma}_{ii}\in (0,1]$ as well as parameters $\gamma,\tau,\zeta,\nu \in (0,1)$ such that 
\begin{itemize}
    \item $| \tilde{\lambda}_i - \lambda_i| \leq \gamma$ with probability at least $1-\tau$, and
    \item $|\tilde{\sigma}_{ij} - \bra{\psi_i} \sigma\ket{\psi_j}| \leq \zeta$ with probability at least $1 - \nu$.
\end{itemize}
Then, for any fixed pair $i <j$, it holds with probability at least $1 -\nu - 2\tau $:
$$
\Big| \sqrt{\lambda_i}\sqrt{\lambda_j} \bra{\psi_i}\sigma\ket{\psi_i} - \sqrt{\tilde{\lambda}_i}\sqrt{\tilde{\lambda}_j}\, \tilde{\sigma}_{ij} \Big| \leq \zeta + \frac{\gamma^2}{2 \kappa^2},
$$
where $\kappa$ is a lower bound on the smallest eigenvalue of $\rho$.
\end{lemma}

\begin{proof}
Let us first bound the additive error of the estimate $\sqrt{\tilde{\lambda}_i} \sqrt{\tilde{\lambda}_j}$. Recall that a simple first-order Taylor expansion for $f(x)=\sqrt{x}$ with additive error $\Delta x$ with respect to $x$ reveals that
$$
\sqrt{x + \Delta x} = \sqrt{x} \cdot \left(1 + \frac{1}{2} \frac{\Delta x}{x}\right) + O(\Delta x^2), \quad\quad \forall x \geq 0.
$$
Using Lagrange's remainder theorem, we can bound the remainder as $\frac{1}{8} \frac{\Delta x^2}{x^2}$.
Assuming that $\kappa \in (0,1)$ is a lower bound on the smallest eigenvalue of $\rho$, we get with probability $1-\tau$,
\begin{align}
|\sqrt{\tilde{\lambda}_i} - \sqrt{\lambda_i} | \leq \frac{\gamma^2}{8 \kappa^2}.
\end{align}
Consequently, with probability at least $1 - 2\tau$, we have
\begin{align}\label{eq:off-diag-evs}
|\sqrt{\tilde{\lambda}_i} \sqrt{\tilde{\lambda}_j} - \sqrt{{\lambda}_i} \sqrt{{\lambda}_j} | \leq \frac{\gamma^2}{4 \kappa^2} + \frac{\gamma^4}{64 \kappa^4} \leq \frac{\gamma^2}{2 \kappa^2}.
\end{align}
where we used the fact that $\gamma \in (0,1)$.
Therefore, by the triangle inequality and \eqref{eq:off-diag-evs}, we obtain the following upper bound with probability at least $1 -\nu - 2\tau $:
\begin{align*}
  &\Big| \sqrt{\lambda_i}\sqrt{\lambda_j} \bra{\psi_i}\sigma\ket{\psi_i} - \sqrt{\tilde{\lambda}_i}\sqrt{\tilde{\lambda}_j}\, \tilde{\sigma}_{ij}\Big| \nonumber \\
&\leq \Big| \sqrt{\lambda_i}\sqrt{\lambda_j} \bra{\psi_i}\sigma\ket{\psi_i} - \sqrt{\lambda_i}\sqrt{\lambda_j}\, \tilde{\sigma}_{ij} \Big| + \Big| \sqrt{\lambda_i}\sqrt{\lambda_j}\, \tilde{\sigma}_{ij} - \sqrt{\tilde{\lambda}_i}\sqrt{\tilde{\lambda}_j}\, \tilde{\sigma}_{ij} \Big|\\
&\leq \sqrt{\lambda_i}\sqrt{\lambda_j} \cdot \big|  \bra{\psi_i}\sigma\ket{\psi_i} -\tilde{\sigma}_{ij} \big| + \tilde{\sigma}_{ij} \cdot \Big| \sqrt{\tilde{\lambda}_i}\sqrt{\tilde{\lambda}_j} - \sqrt{\lambda_i}\sqrt{\lambda_j}\Big|\\
&\leq \zeta + \frac{\gamma^2}{2 \kappa^2}
\end{align*}
This proves the claim.
\end{proof}

\subsection{Bounds for the non-uniform coupon collector problem}\label{sec:CCP}

The non-uniform coupon collector problem was first analyzed by Flajolet et al.~\cite{FLAJOLET1992207} and asks the following: Given a (possibly non-uniform) probability distribution $p = (p_1,\dots,p_n)$ with $p_i > 0 $ over a set of $n$ coupons, where the $i$-th coupon is sampled independently with probability $p_i$, how many independent samples from $p$ are necessary to obtain a full collection? If $p$ is the uniform distribution, the problem is equivalent to the well-known (standard) coupon collector problem.

Let $T_{(p_1,\dots,p_n)}$ be a random variable for the number of samples that are necessary to obtain a complete collection of all $n$ coupons.
Our first result is the following non-trivial upper bound on the average waiting time $\mathbb{E}[T_{(p_1,\dots,p_n)}]$ with respect to the random variable $T_{(p_1,\dots,p_n)}$.

\begin{lemma}\label{lem:non-trivial-CCP}
Let $p = (p_1,\dots,p_n)$ be a distribution such that $p_i > 0 $, for all $i \in [n]$. Then,
$$
\mathbb{E}[T_{(p_1,\dots,p_n)}] = \int_{0}^{\infty} \Big( 1 - \prod_{i=1}^n (1 - e^{- p_i t}) \Big) \,\mathrm{d}t \, \leq \, n \cdot H(p_1,\dots,p_n)^{-1},
$$
where $H(x_1,\dots,x_n) = n/ (\sum_{i=1}^n x_i^{-1})$ is the harmonic mean.
\end{lemma}
\begin{proof}
Using an identity due to Flajolet et al.~\cite{FLAJOLET1992207}, we obtain
\begin{align*}
\mathbb{E}[T_{(p_1,\dots,p_n)}] &= \int_{0}^{\infty} \Big( 1 - \prod_{i=1}^n (1 - e^{- p_i t}) \Big) \,\mathrm{d}t\\
& \leq \int_{0}^{\infty} \Big( \sum_{i=1}^n e^{- p_i t} \Big) \,\mathrm{d}t & (\text{Weierstrass product inequality})\\
& = \sum_{i=1}^n \int_{0}^{\infty}   e^{- p_i t} \,\mathrm{d}t & (\text{switching order of summation and integration})\\
& = \sum_{i=1}^n  \left[- \frac{e^{- p_i t}}{p_i} \right]_0^\infty\\
& = \sum_{i=1}^n p_i^{-1} \quad = \quad n \cdot H(p_1,\dots,p_n)^{-1}. & (\text{by definition})
\end{align*}
\end{proof}

We can now apply the previous lemma in the context of spectral sampling as follows. Let $\rho \in \mathbb{C}^{d \times d}$ be a density matrix with spectrum $\spec(\rho) = (\lambda_1,\dots,\lambda_r)$.
Suppose we have access to a subroutine as in \eqref{eq:QPE-sample} that allows us to sample random pairs of eigenstates and eigenvalues $(\ket{\psi_i},\tilde{\lambda}_i)$ with probability $\lambda_i \in (0,1]$, with the promise that each approximate eigenvalue $\tilde{\lambda}_i$ is sufficiently close to $\lambda_i$ and separated from the remaining spectrum of $\rho$. Let $\mathbb{E}[T_{\spec(\rho)}]$ denote the average number of repetitions needed to obtain a full collection of $r$ distinct eigenvalues.

Applying the inequality in \Cref{lem:non-trivial-CCP} in the context of spectral sampling of a density operator $\rho \in \mathbb{C}^{d \times d}$, we obtain the following upper bound which directly relates the average time to complete the collection to the spectral properties of $\rho$.

\begin{corollary}
Let $\rho \in \mathbb{C}^{d \times d}$ be a density matrix and let $\spec(\rho) = (\lambda_1,\dots,\lambda_r)$ with $r = \rk{\rho}$. Then,
$$
\mathbb{E}[T_{\spec(\rho)}] = \int_{0}^{\infty} \Big( 1 - \prod_{i=1}^r (1 - e^{- \lambda_i t}) \Big) \,\mathrm{d}t \, \leq \, r \cdot H(\spec(\rho))^{-1}.
$$
\end{corollary}
Unfortunately, the above result is not tight. In particular, for the uniform spectrum $(\frac{1}{r},\dots,\frac{1}{r})$, our bound tells us that $\mathbb{E}[T_{(\frac{1}{r},\dots,\frac{1}{r})}]\leq r^2$, whereas a well known result on the (standard) \emph{uniform coupon collector problem} states that the average number of draws is in the order of $\Theta(r \log r)$. 

In order to find an improved bound which asymptotically matches the standard coupon collector result, we use a \emph{coupling argument} which allows us to relate an instance of the non-uniform coupon collector problem to a \emph{worst-case} instance of the uniform coupon collector problem.

To this end, it is convenient to formalize the coupon collector problem with respect to the uniform distribution $\mathcal{U}([0,1])$ on the interval $[0,1]$ as follows.
Let $p = (p_1,\dots,p_n)$ be a (possibly non-uniform) distribution with $p_i > 0 $, for all $i \in [n]$, and consider the following experiment:
\begin{enumerate}
    \item Sample $x \sim \mathcal{U}([0,1])$ according to the uniform distribution over $[0,1]$.
    \item Assign $x$ to a coupon by bucketing it according to the probability distribution $p$.
\end{enumerate}
The problem then becomes: How many samples from $\mathcal{U}([0,1])$ are needed for a full collection? It is easy to see that the experiment is equivalent to the standard (non-uniform) coupon collector problem.\ \\
\ \\
As an example, consider the uniform distribution $p=(\frac{1}{n},\dots,\frac{1}{n})$. In this case, we bucket using the function $f(x) = \lceil x \cdot n\rceil$, which means that we obtain the $i$-th coupon with probability
$$
\underset{{x \sim \mathcal{U}([0,1])}}{\Pr} \left[f(x) = i \right] = \frac{1}{n}.
$$
In the non-uniform case, $p$ implicitly defines a partition of $[0,1]$ since $\sum_{i=1}^n p_i =1$. We let $g$ define the function that maps $x$ to a unique coupon by bucketing it into the respective interval in $[0,1]$.

We prove the following result.

\begin{lemma}[Worst-case bound for the non-uniform coupon collector's problem]\label{lem:couping-coupon-lemma}
Let $p = (p_1,\dots,p_n)$ be a probability distribution with $p_1 \geq p_2 \geq \dots \geq p_n >0$ and let $m$ be the smallest integer such that $p_n \geq \frac{1}{2m}$. Then, we obtain the following upper bound to sample a full collection
$$
\mathbb{E} \big[T_{(p_1,\dots,p_n)} \big] \leq \mathbb{E} \big[T_{(\frac{1}{m},\dots,\frac{1}{m})} \big].
$$
\end{lemma}
\begin{proof}
We use a coupling argument by analyzing the two instances of the coupon collector over the probability measure $\mathcal{U}([0,1])$. By using $f$ and $g$ from before, we can see that the marginal distributions of our coupling match the original instances. In fact, the $i$-th coupon is selected with probability
\begin{align*}
&\underset{{x \sim \mathcal{U}([0,1])}}{\Pr} \left[g(x) = i \right] = 
p_i & (\text{non-uniform distribution}) \\
&\underset{{x \sim \mathcal{U}([0,1])}}{\Pr} \left[f(x) = i \right] = \frac{1}{m}  & (\text{uniform distribution})
\end{align*}
Therefore, our aforementioned coupling over the probability measure $\mathcal{U}([0,1])$ is well-defined. Let $X_t$ be the event that the collection is incomplete at step $t-1$ for the distribution $p$. Similarly, we define $Y_t$ to be the event that the collection is incomplete at step $t-1$ for the distribution $q = (\frac{1}{m},\dots,\frac{1}{m})$.

We now claim that $X_t \subseteq Y_t$, for every $t \geq 1$. In other words, if the collection is incomplete with respect to $p$, then it must be also be incomplete with respect to $q$. Suppose that event $X_t$ occurs, hence there exists a coupon $j$ which has not been collected. Because of the coupling this means that there exists an interval $I_j \subset [0,1]$ for which no $x$ has appeared such that $x \in I_j$. Now, because $\min_{i \in [n]} p_i \geq \frac{1}{2m}$, it follows that $I_j$ contains at least one interval of size $\frac{1}{m}$ in the range of $f$. Hence, there exists a coupon for the distribution $q$ which has not been collected. This proves the claim.

We can now show the following upper bound:
\begin{align*}
\mathbb{E} \big[T_{(p_1,\dots,p_n)} \big]
&= \mathbb{E} \left[\sum_{t=1}^\infty \mathbb{I}_{X_t} \right] & (\text{by definition})\\
&= \sum_{t=1}^\infty \mathbb{E}[\mathbb{I}_{X_t}] & (\text{linearity of expectation})\\
&= \sum_{t=1}^\infty \Pr[X_t]\\
&\leq \sum_{t=1}^\infty \Pr[Y_t] & (\text{using that } X_t \subseteq Y_t)\\
&=  \mathbb{E} \big[T_{(\frac{1}{m},\dots,\frac{1}{m})} \big].
\end{align*}
\end{proof}
	
The following corollary is a simple consequence of Lemma \ref{lem:couping-coupon-lemma}.

\begin{corollary}[Worst-case bound for the thresholded non-uniform coupon collector's problem]\label{cor:threshold-CC}\ \\
Let $p = (p_1,\dots,p_n)$ be an arbitrary probability distribution with $p_1 \geq p_2 \geq \dots \geq p_n >0$ and let $\theta \in (0,1)$ be a threshold value and let $m = \lceil \frac{1}{2\theta} \rceil$. Let $T_{(p_1,\dots,p_n)}^{\theta}$ be the the random variable for the number of repetitions it takes to sample all coupons which occur with probability at least $\theta$. Then,
$$
\mathbb{E} \big[T_{(p_1,\dots,p_n)}^{\theta} \big] \leq \mathbb{E} \big[T_{(\frac{1}{m},\dots,\frac{1}{m})} \big].
$$
\end{corollary}

\subsection{Analysis of the Algorithm}\label{subsec:analysis-spectral-sampling}

Let us now analyze our spectral sampling-based algorithm for (truncated) fidelity estimation.

\begin{theorem}\label{thm:spectral-sampling-theorem}
Let $\rho,\sigma \in \mathbb{C}^{d \times d}$ be arbitrary density matrices, and suppose that $\rho$ has a well-separated spectrum with a gap $\Delta > 0$. Suppose also that $\rho$ has the smallest rank of the two states.
Let $\eps \in (0,1)$, $\delta \in (0,1)$, $\theta \in (0,1)$ be parameters. Then, Algorithm \ref{algorithm:spectral_sampling_fidelity} runs in time
$$
\tilde{O} \left( \frac{T_\rho  +T_\sigma}{\theta^{10.5}\eps^{4} \Delta } + \frac{T_\rho }{\delta \theta^3 \gamma^3} \right)
$$
where $\gamma = \min \{\frac{ \theta^3 \eps}{\sqrt{2}},\frac{\Delta}{2}\}$ and outputs an estimate $\hat{F}(\rho_\theta,\sigma)$ such that with probability $1-\delta$:
$$| F(\rho_\theta,\sigma)-\hat{F}(\rho_\theta,\sigma)| \leq \eps,$$
where $\rho_{\theta}$ is a ``soft-thresholded'' version of $\rho$ in which eigenvalues of $\rho$ below $\theta/2$ are completely removed and those above $\theta$ are kept intact, while eigenvalues in $[\theta/2, \theta]$ are decreased by some amount.
\end{theorem}

\begin{proof}
Let $\eps \in (0,1)$ and $\delta \in (0,1)$ be parameters, and let $\theta \in (0,1)$ be the truncation parameter.
Let $\spec(\rho) = (\lambda_1, \dots, \lambda_{\rk{\rho}})$ be the eigenvalue spectrum of $\rho \in \mathbb{C}^{d \times d}$ and let $\mathrm{rk}_\theta$ be the number of eigenvalues in the interval $(\theta,1]$. Let $m = \lceil \frac{1}{2\theta} \rceil$ be a parameter and note that $\mathrm{rk}_\theta\leq \frac{1}{\theta}$.

We show that it suffices to run \texttt{Quantum\_Spectral\_Sampling}$(U_\rho,\gamma,\ell)$ at most $M = \lceil \frac{16 m H_m}{\delta \cdot \theta^2} \rceil$ times with parameters $\gamma = \min \{\frac{ \theta^3 \eps}{\sqrt{2}},\frac{\Delta}{2}\}$ and $\ell = \lceil \log\big(\frac{16}{\delta \cdot \theta^2}\big) + \log M\rceil$
to find $r_\theta \geq \mathrm{rk}_\theta$ accurate estimates $\tilde{\lambda}_i \geq \frac{3\theta}{4}$ of distinct eigenvalues $\lambda_i \in [\frac{\theta}{2},1]$, including the complete set of eigenvalues in the interval $(\theta,1]$, with probability at least $1-\frac{\delta\cdot \theta^2}{8}$. Let us now analyze the procedure in more detail.

In each iteration of \texttt{Quantum\_Spectral\_Sampling}$(U_\rho,\gamma,\ell)$, we obtain an eigenvalue estimate $\tilde{\lambda}_i$ such that $|\tilde{\lambda}_i - \lambda_i| \leq \gamma$ with probability $1-2^{-l}$. Hence, by the union bound, all $M$ trials will produce accurate estimates with probability at least $1 - \frac{\delta\cdot \theta^2}{16}$.
Moreover, by our choice of parameters, each eigenvalue estimate $\tilde{\lambda}_i$ falls within distance $\gamma < \frac{\Delta}{2}$ of the eigenvalue $\lambda_i$, enabling Algorithm \ref{algorithm:spectral_sampling_fidelity} to uniquely identify each sample. Note also that, since $\gamma < \frac{\theta}{4}$, Algorithm \ref{algorithm:spectral_sampling_fidelity} never accepts any eigenvalue estimates $\tilde{\lambda}_i \geq \frac{3\theta}{4}$ that correspond to eigenvalues $\lambda_i$ in the interval $(0,\frac{\theta}{2})$, thus implying that $r_\theta \leq \frac{2}{\theta}$.

Let us now bound the probability of error, i.e. the probability of not finding a complete set of at least $\mathrm{rk}_\theta$ distinct samples, including all eigenvalues in $(\theta,1]$, after $M$ trials. Let $T_{(\lambda_1, \dots, \lambda_{\mathrm{rk}(\rho)})}^\theta$ be a random variable for the number of samples needed to obtain such a collection.
By Markov's inequality,
\begin{align}\label{eq:time-to-draw-full}
\Pr\left[T_{(\lambda_1, \dots, \lambda_{\mathrm{rk}(\rho)})}^\theta  \geq \frac{16  m H_m}{\delta \cdot \theta^2} \right] \leq \frac{\delta \cdot \theta^2}{16 m  H_m} \mathbb{E}[T_{(\lambda_1, \dots, \lambda_{\mathrm{rk}(\rho)})}^\theta ].
\end{align}
Let $T_{(\frac{1}{m}, \dots, \frac{1}{m})}$ denote the number of repetitions needed to draw $m$ distinct coupons in the uniform coupon collector problem, where $m = \lceil \frac{1}{2\theta} \rceil$. From Corollary \ref{cor:threshold-CC} it follows that
$$\mathbb{E}[T_{(\lambda_1, \dots, \lambda_{\mathrm{rk}(\rho)})}^\theta] \,\leq\, \mathbb{E}[T_{(\frac{1}{m}, \dots, \frac{1}{m})}] = m \cdot H_m,$$ where $H_m$ is the $m$-th harmonic number and $m \cdot H_m = \Theta(m \log m)$.
Hence, the probability in \eqref{eq:time-to-draw-full} is at most $\frac{\delta \cdot \theta^2}{16}$, as required. Therefore, by the union bound, the quantum spectral sampling procedure of Algorithm \ref{algorithm:spectral_sampling_fidelity} succeeds at collecting $r_\theta \geq \mathrm{rk}_\theta$ both distinct and accurate eigenvalue estimates after
$$
M=\left\lceil \frac{16 \cdot m H_m}{\delta \cdot \theta^2}  \right\rceil = \Theta\left( \frac{ \log\big(\theta^{-1}\big)}{\delta \cdot \theta^3}\right)
$$
repetitions of \texttt{Quantum\_Spectral\_Sampling}$(U_\rho,\gamma,\ell)$ with probability at least $1 - \frac{\delta \cdot \theta^2}{8}$. Because $r_0 \leq \frac{2}{\theta}$, this in turn implies that the spectral sampling procedure succeeds with probability at least $1-\frac{\delta}{2r_0^2}$. 

Recall that $\Lambda(\rho,\sigma)=\sqrt{\rho}\sigma\sqrt{\rho} \in \mathbb{C}^{d \times d}$ has the following non-trivial entries in the eigenbasis of $\rho$:
	\begin{align*}
	 \Lambda_{ij}  = \sqrt{\lambda_i}\sqrt{\lambda_j}  \bra{\psi_i} \sigma \ket{\psi_j}, \quad\quad \forall i,j \in \{1,\dots,\rk{\rho}\}.
	\end{align*}
We show that Algorithm \ref{algorithm:spectral_sampling_fidelity} obtains an estimate $\hat{\Lambda}^{\theta}$ of a matrix $\Lambda^\theta = \Lambda(\rho_\theta,\sigma)$, where $\rho_\theta$ is a soft truncation of $\rho$ in which eigenvalues of $\rho$ below $\theta/2$ are completely removed and those above $\theta$ are kept intact, while eigenvalues in $[\theta/2, \theta]$ are potentially incomplete or decreased by some amount.

Let us now consider
the estimates for the diagonal entries of the matrix $\Lambda^{\theta}$ (lines $11$ and $12$). Let $i \in [r_\theta]$ be an index. To estimate $\Lambda_{ii}^\theta$, we first run
\texttt{Swap\_Test}$(U_i \ket{0^{\log d}},\sigma,\xi,\nu)$ with parameters $\xi=\frac{\eps^2}{8r_\theta^4}$ and $\nu = \frac{\delta}{2r_\theta^2}$ to obtain an estimate $\tilde{\sigma}_{ii}$ such that $|\tilde{\sigma}_{ii} - \bra{\psi_i}\sigma\ket{\psi_i}| \leq \xi$ with probability at least $1-\nu$ (by Lemma \ref{lem:approximate-inner-products}), where $U_i =\texttt{Quantum\_Eigenstate\_Filtering}(U_\rho,\tilde{\lambda}_i,\frac{\eps^2}{16r_\theta^4})$ is the unitary in Theorem \ref{thm:quantum-eigenstate-filtering}.
Letting $\hat{\Lambda}_{ii}^{\theta} \leftarrow \tilde{\lambda}_i \tilde{\sigma}_{ii}$, it then follows from Lemma \ref{lem:diagonal-terms} that for every index $i \in [r_\theta]$,
\begin{align}
\Pr\left[\Big|{\Lambda}_{ii}^{\theta} - \hat{\Lambda}_{ii}^{\theta} \Big| \leq \frac{\eps^2}{2r_\theta^4} \right] \geq 1 - \frac{\delta}{r_\theta^2}.
\end{align}

Let us now consider the estimates for the off-diagonal entries of the matrix $\Lambda^\theta$ (lines $14$ to $16$). Let $i,j \in [r_\theta]$ be a pair of indices with $i < j$.
We run \texttt{Hadamard\_Test}$(\ket{0^{\log d}},U_i^\dag U U_j,\xi,\nu)$ with parameters $\xi = \frac{\eps^2}{8r_\theta^4}$ and $\nu = \frac{\delta}{2r_\theta^2}$ to obtain an estimate $\tilde{\sigma}_{ij}$, where
   $U_j =\texttt{Quantum\_Eigenstate\_Filtering}(U_\rho,\tilde{\lambda}_j,\frac{\eps^2}{16r_\theta^4})$ and $U_i^\dag =\texttt{Quantum\_Eigenstate\_Filtering}(U_\rho,\tilde{\lambda}_i,\frac{\eps^2}{16r_\theta^4})^\dag$ are unitary as in Theorem \ref{thm:quantum-eigenstate-filtering}, and where $U$ is a block-encoding of $\sigma$ as in Lemma \ref{lem:blochDensity}.
Letting $\hat{\Lambda}_{ij}^{\theta} \leftarrow \sqrt{\tilde{\lambda}_i}\sqrt{\tilde{\lambda}_j}\, \tilde{\sigma}_{ij}$ and $\hat{\Lambda}^\theta_{ji} = \hat{\Lambda}^{\theta^*}_{ij}$ for $i \neq j$, we then have by Lemma \ref{lem:off-diagonal} that
$$
\Pr\left[\Big|{\Lambda}_{ij}^{\theta} - \hat{\Lambda}_{ij}^{\theta} \Big| \leq \frac{\eps^2}{2r_\theta^4} \right]  \geq 1 - \frac{\delta}{r_\theta^2}.
$$
Given our choice of parameters $\eps$ and $\delta$, we have that for every fixed $i,j \in [r_\theta]$:
$$
\Pr\left[\Big|{\Lambda}_{ij}^{\theta} - \hat{\Lambda}_{ij}^{\theta} \Big| > \frac{\eps^2}{2r_\theta^4} \right]  \leq \frac{\delta}{r_\theta^2}.
$$
Using the continuity bound for fidelity estimation (Theorem \ref{lem:continuity}), we obtain
\begin{align}
|F(\rho_\theta,\sigma) - \hat{F}(\rho_\theta,\sigma) |&=
\big| \mathrm{Tr} \big[{\sqrt{\Lambda^{\theta}}}\big] - \mathrm{Tr}\big[{\sqrt{\hat{\Lambda}_+^{\theta}}}\big] \big|\\
&\leq  \sqrt{2}r_\theta  \cdot \sqrt{\|\Lambda^{\theta} - \hat{\Lambda}^{\theta} \|_1} \nonumber\\
&\leq \sqrt{2}r_\theta^2 \cdot \sqrt{\|\Lambda^{\theta} - \hat{\Lambda}^{\theta} \|_{\max}} \nonumber\\
&= \sqrt{2}r_\theta^2 \cdot \sqrt{\max_{1 \leq i,j \leq r_\theta} |\Lambda_{ij}^{\theta} - \hat{\Lambda}_{ij}^{\theta}|}. \label{eq:continuity-and-norm}
\end{align}
Putting everything together and using the inequality in \eqref{eq:continuity-and-norm}, we find that
\begin{align*}
\Pr\Big[ |F(\rho_\theta,\sigma) - \hat{F}(\rho_\theta,\sigma) | \leq \eps \Big] & = \Pr \Big[ \big| \mathrm{Tr} \big[{\sqrt{\Lambda^{\theta}}}\big] - \mathrm{Tr}\big[{\sqrt{\hat{\Lambda}_+^{\theta}}}\big] \big| \leq \eps \Big] & (\text{by definition})\\
&\geq \Pr\left[\max_{1 \leq i,j \leq r_\theta} \,\Big|\Lambda_{ij}^{\theta} - \hat{\Lambda}_{ij}^{\theta} \Big| \leq \frac{\eps^2}{2r_\theta^4} \right]\\
&\geq \Pr\left[\forall i,\forall j \, : \, \Big|\Lambda_{ij}^{\theta} - \hat{\Lambda}_{ij}^{\theta} \Big| \leq \frac{\eps^2}{2r_\theta^4} \right]\\
&\geq 1 - \sum_{i=1}^{r_\theta} \sum_{j=1}^{r_\theta} \Pr \left[\Big|\Lambda_{ij}^{\theta} - \hat{\Lambda}_{ij}^{\theta}\Big| > \frac{\eps^2}{2r_\theta^4} \right]& (\text{union bound})\\
& \geq  1 - r_\theta^2 \cdot \delta/r_\theta^2 \quad = \quad 1- \delta.
\end{align*}
This proves the claim.
\end{proof}

Finally, we obtain the following result.

\begin{theorem}
Let $\rho,\sigma \in \mathbb{C}^{d \times d}$ be arbitrary density matrices, and suppose that $\rho$ has a well-separated spectrum with a gap $\Delta > 0$ and that $\rho$ has the smallest rank with $\rk{\rho}\leq r$.
Let $\eps \in (0,1)$ and $\delta \in (0,1)$, and fix $\theta =\frac{\eps^2}{4r}$. Then, Algorithm \ref{algorithm:spectral_sampling_fidelity} with parameters $\eps' = \frac{\eps}{2}$, $\delta$ and $\theta$ runs in time
$$
\tilde{O} \left( \frac{r^{10.5}}{\eps^{25} \Delta } \cdot \big(T_\rho  +T_\sigma\big)+ \frac{r^3 T_\rho }{\delta \min \{\frac{\eps^7}{r^3} ,\Delta\}^3} \right)
$$
and outputs an estimate $\hat{F}(\rho_\theta,\sigma)$ such that with probability $1-\delta$:
$$| F(\rho,\sigma)-\hat{F}(\rho_\theta,\sigma)| \leq \eps.$$
\end{theorem}
\begin{proof}
Given the set of parameters $\eps' = \frac{\eps}{2} \in (0,1)$, $\delta \in (0,1)$ and $\theta =\frac{\eps^2}{4r}$,  it follows from Theorem \ref{thm:spectral-sampling-theorem} that Algorithm \ref{algorithm:spectral_sampling_fidelity}
outputs an estimate $\hat{F}(\rho_\theta,\sigma)$ such that with probability $1-\delta$:
$$| F(\rho_\theta,\sigma)-\hat{F}(\rho_\theta,\sigma)| \leq \frac{\eps}{2},$$
where $\rho_{\theta}$ is a ``soft-thresholded'' version of $\rho$ in which eigenvalues of $\rho$ below $\theta/2$ are completely removed and those above $\theta$ are kept intact, while eigenvalues in $[\theta/2, \theta]$ are decreased by some amount.
Therefore, using the soft truncation bound from Corollary \ref{cor:SoftBounds}, we get that with probability $1-\delta$:
\begin{align*}
 |F(\rho,\sigma) - \hat{F}(\rho_\theta,\sigma) | 
 &\leq |F(\rho,\sigma) - F(\rho_\theta,\sigma) | + |F(\rho_\theta,\sigma) - \hat{F}(\rho_\theta,\sigma) | \\
 &\leq \sqrt{ \tr{\Pi_{[0,\theta)} \rho \Pi_{[0,\theta)} } } + |F(\rho_\theta,\sigma) - \hat{F}(\rho_\theta,\sigma) |\\
   &\leq \sqrt{ r \cdot \theta} + |F(\rho_\theta,\sigma) - \hat{F}(\rho_\theta,\sigma)|\\
&\leq \frac{\eps}{2} + \frac{\eps}{2}\\
&= \eps.
\end{align*}
This proves the claim.
\end{proof}

\section{Discussion}
We give an efficient and versatile algorithm for (truncated) fidelity estimation. Our algorithm demonstrates the potential of block-encoding and quantum singular value transformation techniques for quantum information processing tasks. We also demonstrate and work out the specifics of a generic method suggested by~\cite{gilyen2020QAlgForPetzRecovery} for utilizing these techniques in the scenario when one only has access to copies of the states. This method might be of independent interest.

Our alternative spectral-sampling-based algorithm for fidelity estimation performs significantly worse in general compared to our block-encoding algorithm, however it may be easier to implement in certain settings, e.g., when it is easy to obtain circuits that prepare the eigenstates of one of the density operators. For example, consider the problem of exactly simulating the one-dimensional Ising chain~\cite{CerveraLierta2018exactisingmodel}. There, it is possible to efficiently prepare all eigenstates of the Ising Hamiltonian without relying on quantum eigenstate filtering.

\section*{Acknowledgements}
	A.P. would like to thank Ronald de Wolf and Thomas Vidick for many useful suggestions that helped improve the spectral sampling algorithm.

	Part of this work was done while the authors visited the Simons Institute for the Theory of Computing; we gratefully acknowledge its hospitality. 
	
\bibliographystyle{alphaUrlePrint}
\bibliography{qc_gily,LocalBibliography}


\end{document}